\documentclass[acmsmall]{acmart}

\setcopyright{rightsretained}
\acmJournal{PACMMOD}
\acmYear{2024}
\acmVolume{2}
\acmNumber{2 (PODS)}
\acmArticle{112}
\acmMonth{5}
\acmDOI{10.1145/3651613}

\makeatletter
\gdef\@copyrightpermission{
  \begin{minipage}{0.2\columnwidth}
   \href{https://creativecommons.org/licenses/by/4.0/}
   {\includegraphics[width=0.90\textwidth]{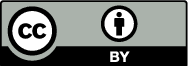}}
  \end{minipage}\hfill
  \begin{minipage}{0.8\columnwidth}
   \href{This work is licensed under a Creative Commons Attribution International 4.0 License.}
  \end{minipage}
  \vspace{5pt}
}
\makeatother

\makeatletter
\gdef\@copyrightpermission{
  \begin{minipage}{0.2\columnwidth}
   \href{https://creativecommons.org/licenses/by/4.0/}{\includegraphics[width=0.90\textwidth]{figs/4acm-cc-by-88x31.eps}}
  \end{minipage}\hfill
  \begin{minipage}{0.8\columnwidth}
   \href{https://creativecommons.org/licenses/by/4.0}{This work is licensed under a Creative Commons Attribution International 4.0 License.}
  \end{minipage}
  \vspace{5pt}
}
\makeatother

\usepackage{lineno}
\usepackage[T1]{fontenc}
\usepackage{graphicx}
\usepackage{amsfonts}
\usepackage{xcolor}
\usepackage{amsmath}
\usepackage[ruled,commentsnumbered,linesnumbered,noend]{algorithm2e}
\usepackage{comment}
\usepackage[inline]{enumitem}
\usepackage{wrapfig}
\usepackage{subcaption}
\usepackage{todonotes}
\usepackage[normalem]{ulem}

\usepackage{thmtools} 
\usepackage{thm-restate}

\usepackage{tikz}
\usetikzlibrary {arrows.meta,automata,positioning}

\usepackage{musicography}

\newenvironment{myparagraph}[1]{\vspace*{0.05in}\noindent{\bf #1.}}{}

\newtheorem{theorem}{Theorem}
\newtheorem{lemma}[theorem]{Lemma}
\newtheorem{claim}[theorem]{Claim}

\newcommand{\figlabel}[1]{\label{fig:#1}}
\newcommand{\figref}[1]{Fig.~\ref{fig:#1}}
\newcommand{\seclabel}[1]{\label{sec:#1}}
\newcommand{\secref}[1]{Section~\ref{sec:#1}}

\newcommand{\thmlabel}[1]{\label{thm:#1}}
\newcommand{\thmref}[1]{Theorem~\ref{thm:#1}}

\newcommand{\lemlabel}[1]{\label{lem:#1}}
\newcommand{\lemref}[1]{Lemma~\ref{lem:#1}}

\newcommand{\itmlabel}[1]{\label{itm:#1}}
\newcommand{\itmref}[1]{\ref{itm:#1}}
\newcommand{\equlabel}[1]{\label{eq:#1}}
\newcommand{\equref}[1]{Equation~(\ref{eq:#1})}
\newcommand{\applabel}[1]{\label{app:#1}}
\newcommand{\appref}[1]{Appendix~\ref{app:#1}}
\newcommand{\algolabel}[1]{\label{algo:#1}}
\newcommand{\algoref}[1]{Algorithm~\ref{algo:#1}}
\renewcommand{\linelabel}[1]{\label{line:#1}}
\renewcommand{\lineref}[1]{Line~\ref{line:#1}}
\newcommand{\linerangeref}[2]{Lines~\ref{line:#1}-\ref{line:#2}}

\newcommand{\claimlabel}[1]{\label{claim:#1}}

\renewcommand{\emptyset}{\varnothing}
\newcommand{\nats}{\mathbb{N}}

\newcommand{\set}[1]{\{#1\}}
\newcommand{\setpred}[2]{\{#1 \,|\, #2\}}

\newcommand{\prob}[1]{{\mathsf{Pr}}\!\left[#1\right]}
\newcommand{\probentangle}[3]{{\mathsf{Pr}}_{#2}\!\left[#1; {#3}\right]}
\newcommand{\probunif}[2]{{\mathsf{Pr}}\!\left[#1; {#2}\right]}
\newcommand{\TVDist}[1]{\mathsf{TV}(#1)}
\newcommand{\TVDistUnif}[2]{{\mathsf{TV}}(#1\;;#2)}
\newcommand{\norm}[1]{\left| #1 \right|}

\newcommand{\emptystr}{\lambda}

\newcommand{\concat}{\cdot}
\newcommand{\alphabet}{\Sigma}
\newcommand{\bitclr}[1]{\textcolor{red!50!black}{\mathtt{#1}}}
\newcommand{\bb}{\bitclr{b}}
\newcommand{\zz}{\bitclr{0}}
\newcommand{\oo}{\bitclr{1}}

\newcommand{\aut}{\mathcal{A}}
\newcommand{\unroll}{\textsf{unroll}}
\newcommand{\autunroll}{\aut_\unroll}
\newcommand{\states}{Q}
\newcommand{\stateq}{q}
\newcommand{\statesP}{P}
\newcommand{\statep}{p}
\newcommand{\init}{I}
\newcommand{\trans}{\Delta}
\newcommand{\accept}{F}
\newcommand{\lang}[1]{\mathcal{L}(#1)}

\newcommand{\epsl}{\epsilon}
\newcommand{\dl}{\delta}
\newcommand{\thresh}{\ensuremath{\texttt{thresh}}}

\newcommand{\pred}[2]{\textsf{Pred}(#1, #2)}
\newcommand{\lvl}{\ell}
\newcommand{\stlvl}[2]{{#1}^{#2}}
\newcommand{\sample}[1]{\mathsf{S}(#1)}
\newcommand{\usample}[1]{\mathsf{U}(#1)}

\newcommand{\approxcnt}[1]{\mathsf{N}(#1)}
\newcommand{\NumSamples}{\mathtt{ns}}
\newcommand{\NumSamplesExtra}{\mathtt{xns}}
\newcommand{\appdel}{\texttt{AppUnion}\xspace}
\newcommand{\sz}{\textsf{sz}}

\newcommand{\Poisson}{\mathsf{Poisson}}

\DontPrintSemicolon
\SetKwInOut{Input}{Input}
\SetKwInOut{Parameters}{Parameters}
\SetKwInOut{Output}{Output}
\SetKwInOut{Initially}{Initially}
\SetKwBlock{Init}{Initialize}{}
\SetKwProg{initblk}{Initialization:}{}{}
\SetKw{Break}{break}
\SetKw{Let}{let}
\SetKw{Skip}{skip}
\SetKwProg{myfun}{function}{}{}
\SetKwFunction{SampleFun}{sample}
\SetKwFunction{AppDelFun}{\appdel\!\!\!$_{(\epsl, \dl)}$}
\SetKw{wp}{w.p.}
\newcommand{\append}[1]{\boldsymbol{\!\cdot}\!\mathsf{append}(#1)}
\newcommand{\contains}[1]{\!\boldsymbol{\cdot}\!\mathsf{contains}(#1)}
\newcommand{\dequeue}{\!\boldsymbol{\cdot}\!\mathsf{deque}()}

\SetKwFor{RepTimes}{repeat}{times:\!}{end}

\makeatletter
\algocf@newcmdside@kobe{SmpBern@prob}{%
    \KwSty{w.p.}\, {#2} :
}
\algocf@newcmdside@kobe{SmpBern@if}{%
    \algocf@block{#2}{}{}%
    \par
}
\algocf@newcmdside{SmpBern@else}{3}{%
    \KwSty{w.p.}\, {#2} :
    \algocf@block{#3}{end}{#3}%
}
\newcommand\SmpBern[4]{%
    \SmpBern@prob{#1}%
    \SmpBern@if{#2}%
    \SmpBern@else{#3}{#4}%
}
\makeatother

\makeatletter
\algocf@newcmdside@kobe{lSmpBern@prob}{%
    \KwSty{w.p.}\, ({#2}) :
}
\algocf@newcmdside@kobe{lSmpBern@if}{%
    {#2}
}
\algocf@newcmdside{lSmpBern@else}{3}{%
\KwSty{ ||\,  w.p.}\, ({#2}) : {#3}
}
\newcommand\lSmpBern[4]{%
    \lSmpBern@prob{#1}%
    \lSmpBern@if{#2}%
    \lSmpBern@else{#3}{#4}%
}
\makeatother

\SetCommentSty{mycommfont}

\newcommand{\ucomment}[1]{\textcolor{red}{#1}}

\newcommand{\failevent}{\mathtt{Fail}}
\newcommand{\smallset}{\textsc{SmallS}}
\newcommand{\accurateEstevent}{\texttt{AccurateN}}
\newcommand{\accurateSmpevent}[1]{{\mathsf{U}}_{#1}}
\newcommand{\accurateEVevent}{\texttt{Accurate}\appdel}
\newcommand{\rtrnsmp}{\mathtt{ret}_{\SampleFun}}
\newcommand{\czero}{\frac{2}{3e}}

\begin{document}

\title{A faster FPRAS for \#NFA}
\titlenote{All authors contributed equally to this research. The symbol \textcircled{r} denotes random author order. The publicly
verifiable record of the randomization is available at \protect\url{https://www.aeaweb.org/journals/policies/random-author-order}}

\author{Kuldeep S. Meel} 
\affiliation{%
  \institution{University of Toronto}
  \city{Toronto}
  \country{Canada}}
\email{meel@cs.toronto.edu}
\orcid{0000-0001-9423-5270}

\author{Sourav Chakraborty} 
\affiliation{%
  \institution{Indian Statistical Institute}
  \city{Kolkata}
  \country{India}}
\email{sourav@isical.ac.in}
\orcid{0000-0001-9518-6204}

\author{Umang Mathur}
\affiliation{%
  \institution{National University of Singapore}
  \city{Singapore}
  \country{Singapore}}
\email{umathur@comp.nus.edu.sg}
\orcid{0000-0002-7610-0660}

\renewcommand{\shortauthors}{Kuldeep S. Meel \lowercase{\textcircled{r}} Sourav Chakraborty \lowercase{\textcircled{r}} Umang Mathur}


\begin{abstract}
	Given a non-deterministic finite automaton (NFA) $\aut$ with $m$ states, and a natural number $n \in \nats$ (presented in unary), the \#NFA problem asks to determine the size of the set $\lang{\aut_n}$ of words of length $n$ accepted by $\aut$. While the corresponding decision problem of checking the emptiness of $\lang{\aut_n}$ is solvable in polynomial time, the \#NFA problem is known to be \#P-hard. Recently, the long-standing open question --- whether there is an FPRAS (fully polynomial time randomized approximation scheme) for \#NFA --- was resolved by Arenas, Croquevielle, Jayaram, and Riveros in \cite{ACJR19}. The authors demonstrated the existence of a fully polynomial randomized approximation scheme with a time complexity of $\widetilde{O}\left(m^{17}n^{17}\cdot \frac{1}{\epsl^{14}}\cdot \log (1/\delta)\right)$, for a given tolerance $\epsl$ and confidence parameter $\delta$.

	Given the prohibitively high time complexity in terms of each of the input parameters, and considering the widespread application of approximate counting (and sampling) in various tasks in Computer Science, a natural question arises: is there a faster FPRAS for \#NFA that can pave the way for the practical implementation of approximate \#NFA tools? In this work, we answer this question in the positive.
	We demonstrate that significant improvements in time complexity are achievable, and
	propose an FPRAS for \#NFA
	that is more efficient in terms of both time and sample complexity.

	A key ingredient in the FPRAS due to Arenas, Croquevielle, Jayaram, and Riveros~\cite{ACJR19}
	is inter-reducibility of sampling and counting, which necessitates a closer look at the more informative measure  --- the number of samples maintained for each pair of state $\stateq$ and length $i \leq n$.
	In particular, the scheme of \cite{ACJR19} maintains 
	$O(\frac{m^7 n^7}{\epsl^7})$ samples per pair of state and length.
	In the FPRAS we propose, we systematically reduce the number of samples 
	required for each state to be only poly-logarithmically dependent on $m$,
	 with significantly less dependence on $n$ and $\epsl$, 
	 maintaining only $\widetilde{O}(\frac{n^4}{\epsl^2})$ samples per state. 
	 Consequently, our FPRAS runs in time 
	 $\widetilde{O}((m^2n^{10}+m^3n^6)\cdot \frac{1}{\epsl^4}\cdot \log^2(1/\delta))$.
	 The FPRAS and its analysis use several novel insights.
	 First, our FPRAS maintains a weaker invariant about the quality of the estimate of the number of samples
	 for each state $\stateq$ and length $i \leq n$.
	 Second, our FPRAS only requires that the distribution 
	 of the samples maintained is close to uniform distribution only in total variation distance (instead of maximum norm).
	We believe our insights may lead to further reductions in time complexity and thus open up a promising avenue for future work towards the practical implementation of tools for approximate \#NFA.
\end{abstract}



\received{December 2023}
\received[revised]{February 2024}
\received[accepted]{March 2024}

\maketitle


\section{Introduction}
\seclabel{intro}

In this paper, we focus on the following computational problem: 

\begin{description}
	\item[\bf\quad\quad\#NFA]: Given an NFA $\aut = (\states, \init, \trans, \accept)$ over the binary alphabet $\alphabet = \{0,1\}$ with $m$ states, and a number $n$ in unary, determine $|\lang{\aut_n}|$, where $\lang{\aut_n}$ is the set of strings of length $n$ that are accepted by $\aut$.
\end{description}

The problem of \#NFA is a fundamental problem in Computer Science, with range of applications in database and information extraction systems;
we outline a few below.

\myparagraph{Probabilistic Query Evaluation}{
	Given a query $Q$ and a database $D$, the problem of probabilistic query evaluation (PQE)
	is to determine the probability of the query holding on a
	randomly sampled database, in which each row is included independently 
	with probability given by its annotated value. 
	PQE is known to be \#P-hard, even in data complexity setting. 
	For a large class of queries,  PQE can 
	directly be reduced to the \#NFA problem.
	In particular, when the schema of the database $D$ consists only of binary relations,
	and when the query $Q$ is a self-join-free path query, then
	the PQE problem for $D$ and $Q$
	reduces to an instance of \#NFA, whose size is linear in the size of $Q$
	and linear in the size of $D$~\cite{vanBremen2023}, and 
	further, the time to compute
	the reduction is also linear in $D$ and $Q$.
	Thus an efficient algorithm for solving \#NFA  directly yields an efficient algorithm for probabilistic query evaluation. 
}

\myparagraph{Counting Answers to Reglar Path Queries}{
	Regular path queries (or \emph{property paths})~\cite{Angles2017} form a rich fragment of graph query languages such as SPARQL designed for graph database systems.
	Here, one abstracts information on edges in the graph database using an alphabet $\alphabet$
	and asks a variety of questions about the set of paths that start from a given source node $u$,
	end at a target node $v$.
	In particular, a regular path query $(u R v)$ asks to 
	enumerate, count or uniformly sample the set of paths (bounded in length
	by some fixed number $n$) from $u$ to $v$
	whose labels match the regular expression $R$.
	Therefore, the  problem of counting the number of answers to the queries reduces to the \#NFA
	problem, for the NFA obtained by the cross product of NFA represented by the database (with $u$ as the intial state and $v$ as the final state) and the NFA that $R$ gets compiled down to.
	This reduced instance is again linear in the size of the database as well as the query,
	and the reduction takes the same time~\cite{ACJR19}.
}

\myparagraph{Probabilistic Graph Homomorphism}{
	A probabilistic graph is a pair $(H, \pi)$, 
	where $H$ is a graph and $\pi$ is the associated probability labeling over edges. 
	The associated probability space is defined as the set of all subgraphs of $H$ 
	wherein every edge $e$ is sampled independently with probability $\pi(e)$.  
	Given a query Graph $G$  and a probabilistic graph $(H,\pi)$ 
	the problem of \emph{probabilistic graph homomorphism problem}
	asks to to compute the probability that a randomly sampled subgraph of $H$
	admits a homomorphism from G. 
	In case of 1-way (and 2-way) path queries, 
	the problem was again shown to reduce to \#NFA~\cite{AvBM23}. 
}

\myparagraph{Beyond databases}{
	From a complexity theoretic viewpoint, \#NFA is SpanL-complete and therefore, 
	finds applications in domains beyond databases such as program verification~\cite{SMC2010}, 
	testing of software~\cite{sutton2007fuzzing} 
	and distributed systems~\cite{Ozkan2019}, 
	probabilistic programming, estimation of information leakage in 
	software~\cite{Bang2016,Gao2019,Saha2023},
	information extraction~\cite{ACJR21},
	automated reasoning using BDDs~\cite{ACJR21},
	and even music generation~\cite{Donze2014}.
}

The direct applications of the \#NFA problem
brings up the central computational question --- how do we count the number of 
elements of $\lang{\aut_n}$? 
It turns out that \#NFA is, in fact, provably computationally intractable --- 
it is known to be a \#P-hard problem~\cite{ALVAREZ19933}. 
Given the intractability, and given the widespread applications of the \#NFA problem, 
the next natural question then 
arises --- can we approximate \#NFA using a polynomial time algorithm?
Until very recently, the only approximation algorithms known for \#NFA
had running times that were \emph{quasi-polynomial} in the size of the inputs~\cite{GoreJKSM97,KannanSM95}.
In a recent breakthrough by Arenas, Croquevielle, Jayaram and Riveros~\cite{ACJR19,ACJR21},
\#NFA was finally shown to admit \emph{fully polynomial time randomized approximation scheme} (FPRAS).

One would assume that the discovery of FPRAS for \#NFA would lead to design of scalable tools for all the different applications outlined above. Such has, however, not been the case and the primary hurdle to achieving practical adoption has been the prohibitively large time complexity of the proposed FPRAS--
it takes time 
$\widetilde{O}\left(m^{17}n^{17}\cdot \frac{1}{\epsl^{14}}\cdot \log (1/\dl)\right)$,
where  $\widetilde{O}$ hides the poly-logarithmic multiplicative factors in $m$ and $n$. 
It is worth highlighting that for most applications,  
the reduction to \#NFA instance often
only incurs a blow up that is a low-degree polynomial (linear or quadratic)
factor  larger than the size
of the original problem. 
This means that the bottleneck to solving the original counting problem
is  really the high runtime complexity involved in the (approximate) 
counting algorithm and not the reduction.

In this work, we pursue the goal of designing a faster algorithm,
and propose an FPRAS which runs in time 
$\widetilde{O}((m^2n^{10} + m^3n^6)\cdot \frac{1}{\epsl^4}\cdot\log^2(1/\dl))$.
The reliance of both FPRASes on the inter-reducibility of sampling and counting necessitates a more informative measure, such as the number of samples maintained for each state. The scheme proposed in ~\cite{ACJR21} maintains $O(\frac{m^7 n^7}{\epsl^7})$ samples per state while the scheme proposed in this work maintains only $O(\frac{n^4}{\epsl^2})$ samples, which is independent of $m$ and with a significant reduction in its dependence on $n$ and $\epsl$.

\subsection{Overview of our FPRAS}

At a high-level, our approach exploits the tight connection between 
approximate counting and the related problem of (almost) uniform generation.
In the context of an NFA $\aut$ and a length $n$, 
the almost uniform generation problem asks
for a randomized algorithm (generator) $G$,
such that,
on each invocation, the probability with which
$G$ returns any given 
string $x \in \lang{\aut_n}$ lies in the range 
$[\frac{1}{|\lang{\aut_n}|}(1 - \epsl), \frac{1}{|\lang{\aut_n}|}(1 + \epsl)]$;
here $\epsl \in (0, 1)$ is a parameter.
Such a random generator is a  polynomial time generator
if its running time is polynomial in $|\aut|$, $n$ and $1/\epsl$.
The existence of a polynomial time almost uniform generator for \emph{self-reducible}~\cite{JerrumVV86} 
problems implies the existence of an FPRAS, and vice versa.

This tight connection between counting and sampling was also the high-level 
idea behind the work of \cite{ACJR21}.
Consequently, our algorithm fits the template procedure of \figref{template2},
also outlined in~\cite{ACJR21}.
For technical convenience, the algorithm unrolls the automaton $\aut$ into an acyclic
graph\footnote{The assumption in  
the template algorithm of \figref{template2} that $\aut$ has only 1 accepting state is without loss of generality.} with $n$
levels; each level $\lvl$ containing the $\lvl^\text{th}$ copy $\stlvl{\stateq}{\lvl}$
of each state $\stateq$ of $\aut$.
A key component in the algorithm is to incrementally compute
an estimate $\approxcnt{\stlvl{\stateq}{\lvl}}$
of the size of the set $\lang{\stlvl{\stateq}{\lvl}}$ 
of words of length $\lvl$ that have a run ending in state $\stateq$.
Next, in order to inductively estimate $\approxcnt{\stlvl{\stateq}{\lvl}}$,
we in turn rely on the estimates of sets 
$\lang{\stlvl{\statep_{\zz}}{\lvl-1}}$ and $\lang{\stlvl{\statep_{\oo}}{\lvl-1}}$;
here $\statep_{\bb}$ is a $\bb$-predecessor state of $\stateq$ ($\bb \in \alphabet = \set{0,1}$), 
i.e., $(\statep_{\bb},\bb, \stateq) \in \trans$ is a transition of $\aut$.
These estimates can be put to use effectively because:
\begin{align*}
 \lang{\stlvl{\stateq}{\lvl}} = \big(\bigcup_{(\statep_{\zz},\zz,\stateq)\in \trans} \lang{\stlvl{\statep_{\zz}}{\lvl-1}}\big) \cdot \set{\zz} \uplus \big(\bigcup_{(\statep_{\oo},\oo,\stateq)\in \trans} \lang{\stlvl{\statep_{\oo}}{\lvl-1}}\big) \cdot \set{\oo}.   
\end{align*}


 \begin{algorithm}[t]
 \captionof{figure}{Algorithm Template: FPRAS for \#NFA}\figlabel{template2}
  \SetInd{0.3em}{0.5em}
  \Input{NFA $\aut$ with $m$ states $\states$ and a single final state $\stateq_\accept \in \states$, $n \in \nats$ (in unary)}
   Unroll the automaton $\aut$ into an acyclic graph 
   $\autunroll$ by making $n+1$ copies $\set{\stateq^\lvl}_{\lvl \in [0, n], \stateq \in \states}$ of the states $\states$ and adding transitions between immediate layers \;
   \ForEach{$\lvl \in \set{0, \ldots, n}, \stateq \in \states$}{
      Compute the estimate $\approxcnt{\stlvl{\stateq}{\lvl}}$ 
      using prior estimates  and samples
      $\set{\approxcnt{\stlvl{\statep}{i}}, \sample{\stlvl{\statep}{i}}}_{i < \lvl, \statep \in \states}$ \;
      Compute a uniform set of samples $\sample{\stlvl{\stateq}{\lvl}}$ using
      $\approxcnt{\stlvl{\stateq}{\lvl}}$ and the prior estimates and samples
      $\set{\approxcnt{\stlvl{\statep}{i}}, \sample{\stlvl{\statep}{i}}}_{i < \lvl, \statep \in \states}$ \;
}
\Return $\approxcnt{\stlvl{\stateq_\accept}{n}}$
\end{algorithm}

We remark that the sets $\set{\lang{\stlvl{\statep}{\lvl-1}}}_{\statep}$ may not be disjoint for 
different predecessor states $\statep$. 
Hence, we cannot estimate $\bigcup_{(\statep_{\bb},\bb,\stateq)\in \trans} \lang{\stlvl{\statep_{\bb}}{\lvl-1}}$ by merely summing the individual estimates 
$\setpred{\approxcnt{\stlvl{\statep_{\bb}}{\lvl-1}}}{(\statep_{\bb},\bb,\stateq)\in \trans}$. To address this, we maintain a set, $\sample{\stlvl{\statep}{\lvl-1}}$, for each $\statep$, which contains samples of strings from $\lang{\stlvl{\statep}{\lvl-1}}$. One possible idea is to
 construct the set $\sample{\stlvl{\statep}{\lvl-1}}$ using strings $w$ that have been sampled uniformly at random from $\lang{\stlvl{\statep}{\lvl-1}}$ but such a process introduces dependence among the sets $\sample{\stlvl{\statep}{\lvl-1}}$  for different  $\statep$. The crucial observation in~\cite{ACJR21} is that for any $k$, the following is true: for any word $w$ of length $k$, all words from $\lang{\stlvl{\stateq}{\lvl}}$ with the suffix $w$ can be expressed as $\big(\bigcup\limits_{\statep \in P} \lang{\stlvl{\statep}{\lvl-k}}\big) \cdot \set{w}$ for some set $P$ of states, where $P$ depends only on $w, \stateq, $ and $k$. This property is referred to as the {\em self-reducible union} property. This insight allows us to utilize the {\em self-reducibility} property for sampling. We can partition $\lang{p^{\lvl-1}}$ into sets of strings where the last character is either $\zz$ or $\oo$, and compute estimates for each partition since each partition can be represented as a union of sets for which we (inductively) have access to both samples and estimates. We then sample the final character in proportion to the size of these estimates. Similarly, we can sample strings, character by character. That is, we can sample strings by progressively extending suffixes of $w$ backwards! Of course, we must also account for the errors that propagate in the estimates of $\approxcnt{\stlvl{\stateq}{\lvl}}$.

Since our algorithm also relies on the {\em self-reducible union} property, 
the basic structure of our algorithm is similar to that 
of~\cite{ACJR21} on high-level. 
There are, however, significant and crucial technical differences that allow us
to obtain an FPRAS with significantly lower time complexity. 
We highlight these differences below by first
highlighting the two key properties underpinning 
the algorithm and ensuing analysis in~\cite{ACJR21}: 
\begin{enumerate}[leftmargin=1.7cm]
	\item[{\bf (ACJR-1).}] For every level $\lvl$, the following condition holds 
	with high probability (here, $\kappa = \frac{nm}{\epsl}$):
\begin{align*}
\begin{array}{l}
   \xi(\lvl) 
   \triangleq\forall \stateq{\in}\states, \forall P{\subseteq}\states, \ \left| 
   \frac{|\sample{\stlvl{\stateq}{\lvl}}\setminus \bigcup\limits_{\statep \in P} \lang{\statep^{\lvl}}|}{|\sample{\stlvl{\stateq}{\lvl}}|}
   -
   \frac{|\lang{\stateq^{\lvl}}\setminus \bigcup\limits_{\statep \in P} \lang{\statep^{\lvl}}|}{|\lang{\stateq^{\lvl}}|} \right| \leq \frac{1}{\kappa^3} 
\end{array}
\end{align*}

\item[{\bf (ACJR-2).}] $\sample{\stlvl{\stateq}{\lvl}}$ is close in $(L_{\infty}$ distance)\footnote{
	In \cite{ACJR21}, the invariant is stated in a different, but equivalent, formulation: conditioned on $ \xi(\lvl)$, the set $\sample{\stlvl{\stateq}{\lvl}}$ follows the distribution of a multi-set obtained by independently sampling every element of $\lang{\stateq^{\lvl}}$ uniformly at random, with replacement.}
 to  the distribution of multi-set obtained by independently  sampling (with replacement) every element of $\lang{\stateq^{\lvl}}$ uniformly at random. 
\end{enumerate}

\noindent
In contrast, the FPRAS we propose maintains the following two weaker invariants: 
\begin{enumerate}[leftmargin=1.6cm]
	\item[{\bf (Inv-1).}] For every state $\stateq$ and level $\lvl$, 
 the event that $\approxcnt{\stlvl{\stateq}{\lvl}} \in  (1 \pm \frac{\epsl}{2n^2})|\lang{\stateq^{\lvl}}|$, which we denote as $\accurateEstevent_{\stateq, \lvl}$, happens with high probability. 
 
 \item[{\bf (Inv-2).}] The distribution of $\sample{\stlvl{\stateq}{\lvl}}$ is 
 close, \emph{in total variation distance}, to the distribution of a multi-set constructed by 
 sampling (with replacement) every element of $\lang{\stateq^{\lvl}}$ uniformly at random. 
\end{enumerate}

Two remarks are in order. 
First,  $\xi(\lvl)$ implies $\accurateEstevent_{\stateq, \lvl}$, and closeness in $L_{\infty}$ is a stringent condition than closeness in total variation distance. 
Second, the technical arguments of ~\cite{ACJR21} crucially rely on their invariants and do not work with the aforementioned weaker arguments.
 Accordingly, this necessitates a significantly different analysis that depends on coupling of the distributions. It's also noteworthy that, due to our reliance on the weaker condition, our method for estimating $\approxcnt{\stlvl{\stateq}{\lvl}}$ differs from that in~\cite{ACJR21}. In fact, our approach is based on an adaptation of the classical Monte Carlo-based computation of the union of sets~\cite{KarpL85}.

For a given state $\stateq$ and subset $P$, achieving a $\frac{1}{\kappa^3}$-additive approximation for $\frac{|\lang{\stateq^{\lvl}}\setminus \cup_{\statep \in P} \lang{\statep^{\lvl}}|}{|\lang{\stateq^{\lvl}}|}$ with a probability of $1-\eta$ requires, according to standard concentration bounds, that the size of $\sample{\stlvl{\stateq}{\lvl}}$ be $O(\kappa^6 \log \eta^{-1})$. To ensure this additive approximation for all states $\stateq$ and subsets $P$, one must consider union bounds over exponentially many events. Specifically, for the analysis to be successful, one must have $\eta^{-1} \in O(2^{mn})$. Consequently, the FPRAS of~\cite{ACJR21} maintains $O(\kappa^7)$ samples for each state $\stateq$. In contrast, for $\accurateEstevent_{\stateq, \lvl}$ to hold with a probability of $1-\eta$, our FPRAS only maintains $O\left(\frac{n^4}{\epsl^2}\right)\log \eta^{-1}$ samples in $\sample{\stlvl{\stateq}{\lvl}}$. Furthermore, a union bound over only $O(mn)$ events is required, thus setting $\eta^{-1} = O(mn)$ is adequate.

The time complexity of
an algorithm that follows the template outlined in~\figref{template2}
(including both the FPRAS of~\cite{ACJR21} as well as the FPRAS developed in this paper), has a quadratic depends on  $|\sample{\stlvl{\stateq}{\lvl}}|$. 
Accordingly, it is instructive to compare the bound of 
$O(\frac{m^7 n^7}{\epsl^7})$ for  $|S(q^{\lvl})|$ in ~\cite{ACJR21} 
to  $O(\frac{n^4}{\epsl^2})$ in our approach. 
Remarkably, the number of samples for every state in our approach is independent of $m$,
the size of the automaton.

\myparagraph{Organization}{
	In \secref{prelim} we recall some helpful background as well as
	introduce a key concept (distribution entanglement) useful for the rest of the paper.
	In \secref{fpras}, we describe our algorithm, together with two auxiliary algorithms
	and state their formal correctness guarantees, whose analysis we
	present in \secref{technical}.
	We conclude in \secref{conclusions}.
}


\section{Preliminaries}
\seclabel{prelim}

	\myparagraph{NFA and words}{
	We consider the binary alphabet  $\alphabet = \set{\zz, \oo}$. 
	All results in this paper are extendable to alphabets of arbitrary but fixed constant size. 
	A string $w$ over $\alphabet$ is either the empty string $\emptystr$ (of length $0$)
	or a non-empty finite sequence $w_1w_2\ldots w_k$ (of length $|w| = k>0$) where each $w_i \in \alphabet$.
	The set of all strings over $\alphabet$ is denoted by $\alphabet^*$.
	A non-deterministic finite automaton is a tuple
	$\aut = (\states, \init, \trans, \accept)$ where $\states$ is a finite set of states,
	$\init \in \states$ is the initial state, 
	$\trans \subseteq \states \times \alphabet \times \states$
	is the transition relation, 
	and $\accept$ is the set of accepting states.
	A run of $w$ on $\aut$ is a sequence
	$\rho = \stateq_0, \stateq_1 \ldots, \stateq_k$ such that
	$k = |w|$, $\stateq_0 = \init$ and for every $i < k$, 
	$(\stateq_i, w_{i+1}, \stateq_{i+1}) \in \trans$;
	$\rho$ is accepting if $\stateq_k \in \accept$.
	The word $w$ is accepted by $\aut$ if there is an accepting run of $w$ on $\aut$.
	The set of strings accepted by $\aut$ is $\lang{\aut}$.
	For a natural number $n \in \nats$, the $n^\text{th}$ slice
	of $\aut$'s language, $\lang{\aut_n}$,
	is the set of strings of length $n$ in $\lang{\aut}$.
	For simplicity, this paper assumes a singleton set of accepting states,
	i.e., $\accept = \set{\stateq_\accept}$ for some $\stateq_\accept \in \states$.
	For $\stateq \in \states$ and $\bb \in \alphabet$, 
	the $\bb$-predecessors of $\stateq$ are given by the set
	$\pred{\stateq}{\bb} = \setpred{\statep}{(\statep, \bb, \stateq )\in \trans}.$
	
	Given automaton $\aut$ and a number $n$ in unary, 
we construct the unrolled automaton $\autunroll$, and assume that all states in $\autunroll$
are reachable from the initial state.
We use the notation $\stlvl{\stateq}{\lvl}$ to denote the $\lvl^\text{th}$
copy of state $\stateq \in \states$. 
 For each state $q$ and level $0\leq \lvl \leq n$, 
we use $\lang{\stlvl{\stateq}{\lvl}}$ to denote the the set of all distinct strings $w$ 
of length $\lvl$ for which there is a path (labeled by $w$) from the starting state to $\stateq$.
}


\myparagraph{Total Variation Distance}{The total variation distance
between two random variables $X,Y$ over domain
$\Omega$  is defined as:
\begin{align*}
	\TVDist{X, Y} = \sum_{\omega \in \Omega} \Pr[X = \omega] - \min(\Pr[X = \omega], \Pr[Y = \omega]).
\end{align*}
}

\subsection{A New Notation: Distribution Entanglement}
\seclabel{distribution-entaglement}
We introduce a new notation, referred to as {\em distribution entanglement}, 
to argue about the probability of events for cases when the distribution of a certain random 
variable were to mimic distribution of another random variable. 
For an event $\mathcal{E}$ and random variables $X$ and $Y$, 
we use $\probentangle{\mathcal{E}}{X}{Y}$ to denote the probability of $\mathcal{E}$ 
when $X$ has \emph{distribution entangled to} $Y$, which is defined as follows:
\begin{align*}
	\probentangle{\mathcal{E}}{X}{Y} = \sum_{\omega \in \Omega} \prob{Y=\omega} \cdot \prob{\mathcal{E} | X=\omega}.
\end{align*}	 
where $\Omega$ is the support of $Y$. The above notion is well-defined only 
when it is the case for all $\omega \in \Omega$, we have $\Pr[X= \omega] >0$. 
Likewise, we can also extend this notion for conditioned events.
That is, for events $\mathcal{E}$, $\mathcal{F}$ and random variables $X$ and $Y$,
the conditional probability of $\mathcal{E}$ given $\mathcal{F}$ when $X$ 
has distribution entangled to $Y$is defined as follows:
\begin{align*}
	\probentangle{\mathcal{E} | \mathcal{F}}{X}{Y} = \sum_{\omega \in \Omega} \prob{Y=\omega} \cdot \prob{\mathcal{E} | \mathcal{F} \cap X=\omega}.
\end{align*}
Whenever $X$ is clear from the context, we will omit it from the notation.



\section{A Faster FPRAS}
\seclabel{fpras}

In this section 
we formally present our FPRAS for \#NFA,
together with the statement of its correctness and running time (\thmref{main}).
Our FPRAS \algoref{main} calls two subroutines \appdel (\algoref{approx-delphic}) 
and \SampleFun (\algoref{sampling}), 
each of which are presented next, together with
their informal descriptions and formal statements of correctness.
\algoref{main} presents the main procedure for approximating the size of 
$\lang{\aut_n}$, where $\aut$ is the NFA and $n$ is the parameter that is given as input. 
The algorithm works in a  dynamic programming fashion
and maintains, for each state $\stateq$ and each length $1 \leq \lvl \leq n$, 
a sampled subset $\sample{\stlvl{\stateq}{\lvl}}$ of $\lang{\stlvl{\stateq}{\lvl}}$. 
As a key step, it uses the sampled sets corresponding to  
$\pred{\stateq}{\zz}$ and $\pred{\stateq}{\oo}$ to construct the desired
sampled subset of $\lang{\stlvl{\stateq}{\lvl}}$. 
We first describe  the  two subroutines the \appdel and \SampleFun.


\subsection{Approximating Union of Sets}

\appdel (\algoref{approx-delphic}) 
takes parameters $\epsl, \dl$ and 
inputs $\epsl_{\sz}$ and (access to)
sets $T_1, \dots, T_k$. 
For each set $T_i$ the access is given in the form of a membership oracle $O_i$,
a subset $S_i$ (obtained by sampling $T_i$ with replacements) presented as a list, 
and an estimate $\sz_i$ of the size of $T_i$. 
Using these, \algoref{approx-delphic} outputs an $(\epsl, \dl)$ 
estimate of the size of the union $\cup_{i=1}^k T_i$.

The precise algorithm presented here represents a modification of the
classic Monte Carlo-based scheme due to Karp and Luby~\cite{KarpL85}.
The proof of its correctness also shares similarities with~\cite{KarpL85}
and we summarize the intuition here.
For a $\sigma \in \cup_{i=1}^{k}T_i$, 
let $L(\sigma)$ denote the smallest index $i \in \set{1, 2, \ldots, k}$
such that $\sigma \in T_{i}$.
Now, let  $\mathcal{U}_\textsf{unique}$ be the set of all pairs 
$(\sigma, L(\sigma))$ and let 
 $\mathcal{U}_\textsf{multiple}$ be the set of all 
 $(\sigma, i)$ pairs such that $\sigma \in T_i$. 
 Consequently, the cardinality of the set $\mathcal{U}_\textsf{unique}$ is 
 $|\cup_{i=1}^k T_i|$,  
 while the cardinality of $\mathcal{U}_\textsf{multiple}$ is $\sum_{i=1}^k |T_i|$. 
 By drawing sufficiently-many ($t$, to be precise) samples from $\mathcal{U}_\textsf{multiple}$ 
 and assessing the fraction belonging to $\mathcal{U}_\textsf{unique}$, 
 one can estimate  $|\mathcal{U}_\textsf{unique}|/|\mathcal{U}_\textsf{multiple}|$, 
 which when multiplied by  $|\mathcal{U}_\textsf{multiple}|$ provides an estimate for the size of 
 $\mathcal{U}_\textsf{unique}$, which is the desired output.


\begin{algorithm}[t]
	\caption{Approximating Union of Sets}
	\algolabel{approx-delphic}
	\SetInd{0.3em}{1em}
	\myfun{\AppDelFun{$\epsl_{\sz}, (O_1, S_1, \sz_1), \ldots, (O_k, S_k, \sz_k)$}}{
		$m \gets   \left\lceil\frac{\sum_j \sz_j}{\max_j \sz_j}\right\rceil  $\; 
		$t \gets \frac{12\cdot (1+\epsl_{\sz})^2 \cdot m}{\varepsilon^2}\cdot \log(\frac{4}{\delta})$ \;
		$Y \gets 0$\;
		\RepTimes{$t$}{ \linelabel{ev-loopstart}
			Sample $i \in \set{1, \ldots, k}$ with  probability $p_i = \frac{\sz_{i}}{\sum_j \sz_j}$ \linelabel{ev-samplei}\;
			\lIf{$S_i \neq \emptyset$}{ \linelabel{ev-if}
					$\sigma \gets S_{i}\dequeue{}$ \linelabel{ev-samplesigma}
			}\lElse{
				\Break \linelabel{ev-break}
			}
			\lIf{$\neg(\exists j < i, O_j\contains{\sigma})$}{ \linelabel{ev-check}
				$Y \gets Y + 1$ \linelabel{ev-loopend}
			}
		}
	\Return $\frac{Y}{t} \cdot ({\sum_i \sz_i})$
	}
\end{algorithm}



After sampling an index $i$, in \lineref{ev-samplesigma}
an element is chosen (and removed) from the sampled list $S_i$. 
In the case when the set $S_i$ is constructed by uniformly selecting elements from $T_i$,
this step mimics drawing a random sample from the actual set $T_i$. 
Initially, each set (more precisely, list) $S_i$ is guaranteed to have 
a size of at least $\thresh = 24 \cdot \frac{(1+\epsl_{\sz})^2}{\epsl^2}\cdot \log(\frac{4k}{\delta})$, surpassing the expected number of samples required from $T_i$ during the 
$t$ iterations of the loop. 
If the algorithm ever needs more samples than $|S_i|$ during the $t$ iterations, 
it counts as an error, but the probability of this occurrence will be demonstrated to be very low.
Finally, \lineref{ev-check} verifies whether the sample from $\mathcal{U}_\textsf{multiple}$ 
belongs to  $\mathcal{U}_\textsf{unique}$ by asking a membership question to the oracle $O_i$.
If all the sampled sets $S_i$'s are constructed uniformly,
then the variable  $Y$ tallies the number of samples from $\mathcal{U}_\textsf{unique}$,
and after $t$ iterations, $Y/t$ provides a estimation for 
$|\mathcal{U}_\textsf{unique}|/|\mathcal{U}_\textsf{multiple}|$.  
The final output is the product of this value with $(\sum_i \sz_i)$ (which is the estimate for $|\mathcal{U}_\textsf{multiple}|$).

We next present (in \thmref{approx-delphic-correct})
the formal correctness statement of this algorithm using distribution entanglement 
(see \secref{distribution-entaglement}).
For this, we set the source random variable $X$
to be the one that corresponds to the product distribution
of the input sampled sets $S_1, \ldots, S_k$,
while the target random variable $Y$ corresponds to the product 
distribution of $\accurateSmpevent{1}, \ldots, \accurateSmpevent{k}$,
where, $\accurateSmpevent{i}$, is the 
random variable that obeys the distribution obtained when constructing
a subset of $T_i$ (of at least $\thresh$ many elements) by uniformly picking each element of $T_i$.
Since the random variables $S_1, \ldots, S_k$ are clear from context,
we will use $\probunif{\mathcal{E}}{\accurateSmpevent{1, \ldots, k}}$ 
to denote the probability of event $\mathcal{E}$ in the resulting  entanglement distribution.
For well definedness, we require that $\prob{S_i = \omega} > 0$
for every subset $\omega \subseteq T_i$ (with $\omega \geq \thresh$),
and will ensure this each time we invoke this statement.

\begin{theorem}
\thmlabel{approx-delphic-correct}
Let $\epsl, \dl > 0$, and let $\Omega$ be some set.
Let $T_1, T_2, \ldots, T_k \subseteq \Omega$ be sets with membership oracles
$O_1, \ldots, O_k$ respectively.
Let $\epsl_\sz \geq 0$ and
let $\sz_1, \ldots, \sz_k \in \nats$ be such that for every $i \leq k$, 
we have $\frac{\norm{T_i}}{(1+\epsl_\sz)}\leq \sz_i \leq (1+\epsl_\sz) \norm{T_i}$.
Let $\thresh = 24 \cdot \frac{(1+\epsl_{\sz})^2}{\epsl^2}\cdot \log(\frac{4k}{\delta})$.
For each $i \leq k$, let $S_i$ be some sequence of samples 
(with possible repetitions), obtained according to some distribution
such that $|S_i| \geq \thresh$
and for every $\omega \subseteq T_i$ with $|\omega| \geq \thresh$, $\prob{S_i = \omega} > 0$.
Then, on input $\epsl_{\sz}, (O_1, S_1, \sz_1), \ldots, (O_k, S_k, \sz_k)$,
the algorithm \appdel with parameters $\epsl, \dl$
(\algoref{approx-delphic}) outputs $\sz$ that satisfies:
\begin{align*}
\equlabel{EVCorrectness}
 \probunif{\frac{|\cup_{i=1}^k T_i|}{(1+ \epsl)(1 + \epsl_{\sz})} \leq \norm{\sz} < (1 + \epsl)(1+\epsl_{\sz})|\cup_{i=1}^k T_i| \,}{\accurateSmpevent{1, \ldots, k}} > 1 - \dl.
\end{align*}
The  algorithm makes $O\left(k\cdot (1+\epsl_{\sz})^2 \cdot \frac{1}{\epsl^2} \cdot \log\big(\frac{k}{\dl}\big)\right)$ number of  membership calls to the oracles $O_1, \ldots, O_k$ in the 
worst case  and the rest of the computation takes a similar amount of time.
\end{theorem}

\subsection{Sampling Subroutine}

The sampling subroutine (\algoref{sampling}) takes as input a number $\lvl$, 
a set of states $\stlvl{P}{\lvl}$ (at level $\lvl$ of the unrolled automaton $\autunroll$), 
a string $w$ (of length $n-\lvl$), 
a real number $\varphi$ (representing a probability value), 
an error parameter $\beta$ and a confidence parameter $\eta$,
and outputs a string sampled from the language 
$\cup_{\stateq \in \stlvl{P}{\lvl}} \lang{\stlvl{\stateq}{\lvl}}$.  
\algoref{sampling} is a recursive algorithm.
In the base case ($\lvl=0$), it outputs the input string $w$ with probability $\varphi$, 
and with the remaining probability, it outputs $\bot$. 
In the inductive case ($\lvl > 0$), the algorithm computes estimates 
$\set{\sz_{\bb}}_{\bb \in \set{\zz, \oo}}$,
where $\sz_{\bb}$ is the estimate of the size of 
$\cup_{\statep_{\bb,i}\in \pred{\stateq}{\bb}} \lang{\stlvl{\statep_{\bb,i}}{\lvl-1}}$.
For this, it calls $\AppDelFun$ with previously computed estimates
of states at previous levels $\set{\approxcnt{\stlvl{\stateq}{\lvl-1}}}_{\stateq \in \states}$
and sampled subsets of strings $\set{\sample{\stlvl{\stateq}{\lvl-1}}}_{\stateq \in \states}$,
and uses the unrolled automaton $\autunroll$ as a proxy for membership oracles (see the last argument
to $\AppDelFun$ in \lineref{sample-appunion-call}).
Next, $\bb$ is chosen randomly to be $\zz$ or $\oo$ with
probability proportional to the estimates 
$\sz_{\oo}$ and $\sz_{\zz}$ respectively (\lineref{sampling7}). 
Once $\bb$ is chosen the set $P^{\lvl}$ is updated to $P_{\bb}^{\lvl-1}$ 
and the string $w$ is updated to $\bb\cdot w$. 
Subsequently, the \SampleFun{}
subroutine is called recursively 
with the updated subset of states and the updated $w$. 


\begin{algorithm}[t]
\caption{Sampling subroutine for $\autunroll$, length $n$}
\algolabel{sampling}
\SetInd{0.3em}{1em}
\myfun{\SampleFun{$\lvl$, $\stlvl{\statesP}{\lvl}$, $w$, $\varphi$,$\beta$, $\eta$}}{
$\eta' \gets \eta/4n$\; 
$\beta' \gets (1+\beta)^{\ell-1} - 1$\;
\If{$\lvl = 0$}{
    \lIf{$\varphi > 1$}{\Return $\bot$} \linelabel{overflow}
    \lSmpBern{$\varphi$}{\Return $w$}{$1-\varphi$}{\Return $\bot$ \linelabel{smallprob}}
}
\Else{
    \ForEach{$\bb \in \set{\zz, \oo}$}{
            $\stlvl{\statesP}{\lvl}_{\bb} \gets \bigcup_{\statep \in \stlvl{\statesP}{\lvl}} \pred{\statep}{\bb}$ \;
            \Let $\stlvl{\statesP}{\lvl}_{\bb}$ be the set $\set{\statep_{\bb,1}, \ldots, \statep_{\bb,k_{\bb}}}$ \;
            $\sz_{\bb} \gets \appdel_{\beta,\eta'}\big($
            $\beta', (\sample{\stlvl{\statep}{\lvl-1}_{\bb,1}}, \approxcnt{\stlvl{\statep}{\lvl-1}_{\bb,1}}), \ldots, (\sample{\stlvl{\statep}{\lvl-1}_{\bb,k_{\bb}}}, \approxcnt{\stlvl{\statep}{\lvl-1}_{\bb,k_{\bb}}}),\autunroll$ 
            $\big)$ \linelabel{sample-appunion-call}
        }
    $\mathsf{pr}_{\zz} \gets \frac{ \displaystyle \sz_{\zz} }{\displaystyle \sz_{\zz} + \sz_{\oo} }$ \;
    \lSmpBern{$\mathsf{pr}_{\zz}$}{$\bb \gets \zz$}{$1-\mathsf{pr}_{\zz}$}{$\bb \gets \oo$} \; \label{line:sampling7}
    $\stlvl{\statesP}{\lvl-1} \gets \stlvl{\statesP}{\lvl}_{\bb}$ \;
    $w \gets \bb \concat w$
}
 \Return \SampleFun{$\lvl-1$, $\stlvl{\statesP}{\lvl-1}$, $w$, $\varphi/\mathsf{pr}_{\bb}, \beta, \eta$}
 }
\end{algorithm}

We next tread towards formally characterizing the 
guarantee of \algoref{sampling}. To this end, we will introduce few notations: 
\begin{itemize}[left=0cm]
    \item
    For each $1\leq i \leq \lvl$ and each $\stateq \in \states$,
    let $\accurateEstevent_{i,\stateq}$ be the event 
    that
    $\frac{|\lang{\stlvl{\stateq}{i}}|}{(1+\beta)^{i}} \leq \approxcnt{\stlvl{\stateq}{i}}  \leq (1+\beta)^{i} |\lang{\stlvl{\stateq}{i}}|$.
    Also, let $\accurateEstevent_{i} = \bigcap\limits_{\stateq \in \states} \accurateEstevent_{i, \stateq}$
    and $\accurateEstevent_{\leq i} = \bigcap\limits_{1\leq j \leq i} \accurateEstevent_{j}$.

    \item
    For a state $\stlvl{\stateq}{\lvl}$,
    we use $\usample{\stlvl{\stateq}{\lvl}}$ 
    to represent the random variable denoting the sequence of 
    samples (with possible repetitions) constructed by 
    repeatedly sampling 
    $|\sample{\stlvl{\stateq}{\lvl}}|$ number of elements 
    uniformly from $\lang{\stlvl{\stateq}{\lvl}}$.
    Let ${\sf S}_{\leq i}$ be the random variable corresponding to the product of 
    $\set{\sample{\stlvl{\stateq}{j}}}_{\stateq \in \states, 0 \leq j \leq i}$.
    Also, let $\accurateSmpevent{\leq i}$ represent the random variable 
    obtained by taking the ordered profuct of all $\set{\usample{\stlvl{\stateq}{j}}}_{\stateq \in \states, 0 \leq j \leq i}$.
    In our algorithms, we will use 
    $\probunif{\mathcal{E}}{\accurateSmpevent{\leq i}}$
    (resp. $\probunif{\mathcal{E} | \mathcal{F}}{\accurateSmpevent{\leq i}}$)
    to denote the probability (resp. conditional probability)
    of $\mathcal{E}$ (resp. $\mathcal{E}$ conditioned on $\mathcal{F}$)
    when ${\sf S}_{\leq i}$ is distribution entangled to $\accurateSmpevent{\leq i}$.


\item Consider the random trial corresponding to the call to the procedure
$\SampleFun(\lvl, \set{\stlvl{\stateq}{\lvl}}, \emptystr, \gamma_0, \beta, \eta)$ for $\lvl > 0$
and let $\rtrnsmp$ denote the random variable representing the return value thus obtained.
Observe that, in each such trial, the function $\appdel$ is called $2\lvl$ times, 
twice for each level $1 \leq i \leq \lvl$.
Let us now define some events of interest.

\item 
    For each $1\leq i \leq \lvl$ and $\bb \in\set{\zz,\oo}$, by
    $\accurateEVevent_{i, \bb}$ we denote the event that the 
    $\bb^\text{th}$ call to $\appdel$
    corresponding to level $i$ returns a value $\sz^i_{\bb}$ that satisfies
    $\frac{1}{(1+\beta)^i} \cdot |\lang{P^i_{\bb}}| 
    \leq 
    \sz^i_{\bb} 
    \leq 
    (1+\beta)^i \cdot |\lang{P^i_{\bb}}|$. 
    Here, $P^i$ is the argument of the call at level $i$,
    and $P^i_{\bb}$ is the $\bb$-predecessor of $P^i$. 
    For ease of notation, we also define 
    $\accurateEVevent = \bigcap_{
        \begin{subarray}{l}
        1 \leq i\leq \lvl \\ 
        \bb \in \{\zz, \oo\} 
        \end{subarray}
    }\accurateEVevent_{i, \bb} $
    
 \item 
    Let $\failevent_1$ be the event that $\rtrnsmp=\bot$
    is returned because $\varphi > 1$ at the time of return (\lineref{overflow}).
    Let $\failevent_2$ be the event that $\rtrnsmp=\bot$
    is returned at \lineref{smallprob}.
    Finally, let $\failevent = \failevent_1 \uplus \failevent_2$.

\end{itemize}

The following presents the correctness of~\algoref{sampling},
while its proof is presented in~\secref{samplingcorrectness}.
\begin{theorem}
\thmlabel{correctness-sampling-subroutine}

Let $\gamma_0 = \czero \cdot \frac{1}{\approxcnt{\stlvl{\stateq}{\lvl}}}, \beta \leq \frac{1}{4\cdot n^2}$ and $\set{\sample{\stlvl{\statep}{i}}}_{1\leq i < \lvl, \statep \in \states}$ satisfy     
   \[\forall\; 1 \leq i \leq \lvl, \forall \statep \in \states, \;\;  
     |\sample{\stlvl{\statep}{i}}| \geq \frac{24 \cdot e}{\beta^2} \log \left(\frac{16mn}{\eta}\right).\]
 
Then for the trial corresponding to the invocation 
$\SampleFun(\lvl, \set{\stlvl{\stateq}{\lvl}}, \emptystr, \gamma_0, \beta, \eta)$
the following hold
\begin{enumerate}[label=(\arabic*)]
    \item\itmlabel{uniform-on-support}
    For each $w \in \lang{\stlvl{\stateq}{\lvl}}$,
        $\prob{ \rtrnsmp = w \;\; \big\vert 
                \accurateEVevent \cap \accurateEstevent_{\lvl, \stateq}
                 } = \gamma_0$
    
    
    \item\itmlabel{failure-given-accurate}
$\prob{ \failevent \mid   \accurateEVevent \cap \accurateEstevent_{\lvl, \stateq} }
    \leq 1 - \frac{2}{3e^2}
$
    
    \item\itmlabel{accurate-given-good}
 $       \probunif{\accurateEVevent  \mid 
        \accurateEstevent_{\leq \lvl-1}}{\accurateSmpevent{\leq \lvl-1}}
        \geq 1-\frac{\eta}{2}$
\end{enumerate}
\end{theorem}


\subsection{Main Algorithm}

\begin{algorithm}[t]
\caption{Algorithm for estimating $|\lang{\aut_n}|$}
\algolabel{main}
\SetInd{0.3em}{1em}
\Input{NFA $\aut$ with $m$ states, $n \in \nats$  in unary}
\Parameters{Accuracy parameter $\varepsilon$ and confidence parameter $\delta$}

$\beta \gets \frac{\varepsilon}{4n^2}; \eta \gets \frac{\delta}{2\cdot n \cdot m}$\; 
$\NumSamples \gets \frac{4096 \cdot e \cdot n^4 }{\epsilon^2}\log \left(\frac{4096\cdot m^2 \cdot n^2\cdot \log (\varepsilon^{-2})}{\delta}\right)$\; 
$\NumSamplesExtra \gets \NumSamples \cdot 12\cdot(1-\frac{2}{3e^2})^{-1}\cdot \log (8/\eta)$\; 
Construct the labeled directed acyclic graph $\autunroll$ from $\aut$  \; \linelabel{line1}
\For{$\lvl \in \set{0, \ldots, n} \text{ and } \stateq \in \states$}{
    \If{$\lvl = 0$}{\linelabel{line2}
    	\If{$\stateq  = \init$}{
    	$\big(\approxcnt{\stlvl{\stateq}{\lvl}}, \sample{\stlvl{\stateq}{\lvl}}\big) \gets (1, {\emptystr})$
    }
    	\Else{
    	$\big(\approxcnt{\stlvl{\stateq}{\lvl}}, \sample{\stlvl{\stateq}{\lvl}}\big) \gets \big(0,\emptyset)$ \linelabel{line3}
    }
    }
    \Else{
        \ForEach{$\bb \in \set{\zz, \oo}$}{\linelabel{line4}
            \Let $\set{\statep_{\bb,1}, \ldots, \statep_{\bb,k_{\bb}}}$ be the set $\pred{\stateq}{\bb}$ \;
            $\beta' \gets (1 + \beta)^{\ell -1} -1$\;
            $\sz_{\bb} \gets \appdel_{\beta,\frac{\eta}{2}  \cdot(1  - \frac{1}{2^{n+1}})}\bigg(\beta',
            (\sample{\stlvl{\statep_{\bb,1}}{\lvl-1}},
            \approxcnt{\stlvl{\statep_{\bb,1}}{\lvl-1}}),
            \ldots, 
            (\sample{\stlvl{\statep_{\bb,k_{\bb}}}{\lvl-1}}, \approxcnt{\stlvl{\statep_{\bb,k_{\bb}}}{\lvl-1}}),
            \autunroll\bigg)$ \linelabel{line6}
        }
        \SmpBern{
            $1 - \frac{\eta}{2^{n}}$
        }{
            $\approxcnt{\stlvl{\stateq}{\lvl}} \gets \sz_{\zz} + \sz_{\oo}$
        }{
            $\frac{\eta}{2^{n}}$
        }{
            $\approxcnt{\stlvl{\stateq}{\lvl}} \xleftarrow{\text{Unif-smp}} \set{0, 1, \ldots, 2^{\lvl}} $
        }
        $\sample{\stlvl{\stateq}{\lvl}} \gets \emptyset$ \;
        \RepTimes{$\NumSamplesExtra$}{\linelabel{line7}
            \If{$|\sample{\stlvl{\stateq}{\lvl}}| < \NumSamples$}{
                $w \gets$ \SampleFun{$\lvl, \set{\stlvl{\stateq}{\lvl}}, \emptystr, \czero\cdot\frac{1}{\approxcnt{\stlvl{\stateq}{\lvl}}}, \beta, \frac{\eta}{2\cdot \NumSamplesExtra}$} \; 
                \If{$w \neq \bot$}{
                $\sample{\stlvl{\stateq}{\lvl}} \gets \sample{\stlvl{\stateq}{\lvl}}\append{w}$}\linelabel{line10}
            }\linelabel{line11}
            \lElse{\Break}
        }
        \Let $\alpha_{\stlvl{\stateq}{\lvl}} = \NumSamples{-}|\sample{\stlvl{\stateq}{\lvl}}|$ \linelabel{padding-start} \;
        \If{$\alpha_{\stlvl{\stateq}{\lvl}} > 0 $}{
        \Let $w_{\stlvl{\stateq}{\lvl}}$ be some word in $\lang{\stlvl{\stateq}{\lvl}}$ \;
        \lRepTimes{$\alpha_{\stlvl{\stateq}{\lvl}}$}{
            $\sample{\stlvl{\stateq}{\lvl}}{\gets}\sample{\stlvl{\stateq}{\lvl}}\append{w_{\stlvl{\stateq}{\lvl}}}$ \linelabel{padding-end}
            }
        }
    }
}
\Return $\approxcnt{\stlvl{\stateq}{n}_{\accept}}$
\end{algorithm}

\algoref{main} first constructs (in~\lineref{line1})
the labeled directed acyclic graph $\autunroll$. 
The algorithm then goes over all the states and all
levels $0 \leq \lvl \leq n$
to construct the sampled sets and the size estimates inductively.
In~\linerangeref{line2}{line3}, the algorithm caters for the
base case ($\lvl = 0$). 
In the inductive case (\linerangeref{line4}{line6})
the algorithm computes the estimate $\sz_{\oo}$ and $\sz_{\zz}$, 
where $\sz_{\bb}$ is the estimate of the size of 
$\cup_{\statep_{\bb,i}\in \pred{\stateq}{\bb}} \lang{\stlvl{\statep_{\bb,i}}{\lvl-1}}$. 
The size estimates of, and
a sample subsets of the sets $\set{\lang{\stlvl{\stateq}{\lvl-1}}}_{\stateq \in\states}$ 
are available inductively at this point.
Further, membership in these languages can be easily checked.
As a result, the algorithm uses the subroutine \appdel to compute the size 
of the union (in \lineref{line6}), with $\autunroll$ as the membership oracle as before. 
Once the size estimates $\sz_{\zz}$ and $\sz_{\oo}$ are 
determined, the algorithm constructs the sampled subset of 
$\lang{\stlvl{\stateq}{\lvl}}$ by making $\mathsf{xns}$ 
calls to the subroutine $\SampleFun$ (\linerangeref{line7}{line11});
if fewer than $\NumSamples$ many strings were sampled,
the algorithm simply pads one fixed string from $\lang{\stlvl{\stateq}{\lvl}}$
so that $\sample{\stlvl{\stateq}{\lvl}}$ eventually has size $\NumSamples$ (see \linerangeref{padding-start}{padding-end}).
\thmref{main} presents the desired correctness statement of our FPRAS
~\algoref{main}, and its proof is presented
in~\secref{mainalgo}.

\begin{theorem}
\thmlabel{main}
Given an NFA $\aut$ with $m$ states and  $n \in \nats$ (presented in unary), 
\algoref{main} returns $\mathsf{Est}$ such that the following holds:
\begin{align*}
    \prob{\frac{|\lang{\aut_n}|}{1+\epsl} \leq \mathsf{Est} \leq (1+\epsl) |\lang{\aut_n}|   } \geq 1-\dl
\end{align*}
Moreover, the algorithm has time complexity 
$\widetilde{O}((m^2n^{10} + m^3n^6)\cdot \frac{1}{\epsilon^4}\cdot\log^2(1/\delta))$, where the tilde hides for polynomial factors of $\log (m+n)$.
\end{theorem}



\section{Technical Analysis}
\seclabel{technical}

We now present the detailed technical analyses of the algorithms presented in \secref{fpras}. 


\subsection{Correctness of~\algoref{approx-delphic}}
\applabel{approxunion}


The proof is in the same lines as that of Karp and Luby \cite{KarpL85}.
For the sake of completeness, we present a full proof. 

\begin{proof}[Proof (of \thmref{approx-delphic-correct})]
Let $\mathsf{Fail}$ be the event that the output of the algorithm 
is not within $\frac{|\cup_{i=1}^k T_i|}{(1+\epsl)(1+\epsl_{\sz})}$ and 
$|\cup_{i=1}^k T_i|{(1+\epsl)(1+\epsl_{\sz})}$. 
We will show that $\probunif{\mathsf{Fail}}{\accurateSmpevent{1,\ldots, k}}$ 
is upper-bounded by $\dl$.

Note that the algorithm in each of the $t$ runs of the loop 
(\linerangeref{ev-loopstart}{ev-loopend}) 
tries to draw an element from some $S_i$.
So if we assume that the size of the sets $S_i$ for all $i$ is more than $t$, 
then the condition in the \textbf{if} (in \lineref{ev-if}) will be satisfied 
and in that case the \textbf{else} clause in \lineref{ev-break} is redundant. 

We will prove the correctness of ~\algoref{approx-delphic} in two parts. 
In the first part we will prove that if the size of the sets $S_i$ 
is greater than $t$ then $\probunif{\mathsf{Fail}}{\accurateSmpevent{1,\ldots, k}} \leq \frac{\delta}{2}$. 
In the second part we will prove that if the sets $S_i$ has size {\thresh} 
(much smaller than $t$) then the probability that the algorithm will 
ever need the \textbf{else} clause in \lineref{ev-break} is less that $\delta/2$. 
The two parts combined proves that if the sets $S_i$s are of size 
{\thresh}, then $\probunif{\mathsf{Fail}}{\accurateSmpevent{1,\ldots, k}} \leq \delta$ 
which is what is claimed in the theorem. 

 For all $\sigma \in \cup_{i=1}^{k}T_i$ let $L(\sigma)$ denote 
 the smallest index $i$ such that $\sigma\in T_{i}$. In other words, 
$\sigma$ is in $T_{L(\sigma)}$ and for all $i< L(\sigma)$, $\sigma \not\in T_i$.  
Observe  $\mathcal{U}_{\textsf{unique}} :=  \left\{ (\sigma, L(\sigma))\ \mid\ \sigma\in \cup_{i=1}^k T_i\right\}$ 
has size exactly $\left|  \cup_{i=1}^k T_i \right|$.  
Furthermore, $\mathcal{U}_{\textsf{multiple}} := \left\{ (\sigma, i)\ \mid\ \sigma \in T_i\right\}$ 
has size $\sum_{i=1}^k |T_i|$. 

\

\textbf{Part 1: If $\forall i$,  $|S_i| \geq t$ then $\probunif{\mathsf{Fail}}{\accurateSmpevent{1,\ldots, k}} \leq \frac{\delta}{2}$.} 
If $|S_i| \leq t$ for all $i$,
then the condition in the \textbf{if} clause in \lineref{ev-if} 
is always satisfied. 
For a run of the loop (\linerangeref{ev-loopstart}{ev-loopend}), 
consider the \linerangeref{ev-samplei}{ev-samplesigma}. 
We say the pair $(\sigma, i)$ is sampled in round $j$ 
if in the $j^\text{th}$ iteration of the loop 
$i$ is sampled in \lineref{ev-samplei} and then 
in that same iteration $\sigma$ 
is obtained from $S_i.\mathsf{dequeue()}$ in the \lineref{ev-samplesigma}. 

For any $i_0 \in \set{1, \ldots, k}, \sigma_0\in S_{i_0}$ and $j_0 \in \set{1, \ldots, t}$, 
we have
\begin{equation}
 \probunif{(\sigma_0, i_0) \mbox{ is sampled in round }j_0}{\accurateSmpevent{1,\ldots, k}}  =  
\begin{aligned}
\begin{array}{l}
 \quad \probunif{i_0  \mbox{ is sampled in \lineref{ev-samplei} in round } j_0}{\accurateSmpevent{1,\ldots, k}} \\
\times \prob{\begin{aligned}\sigma_0 \mbox{ is obtained in \lineref{ev-samplesigma} in round} \\j_0 \mid i_0 \mbox{ is sampled in \lineref{ev-samplei}} ; \accurateSmpevent{1,\ldots, k} \end{aligned} } 
 \end{array}
\end{aligned}
\equlabel{appunionone}
\end{equation} 
Since the choice of the set to pick $i_0$ is independent of the distributions of any of the 
$S_i$'s,
we have,
\[\probunif{i_0  \mbox{ is sampled in \lineref{ev-samplei}}}{\accurateSmpevent{1,\ldots, k}} = \Pr[i_0  \mbox{ is sampled in \lineref{ev-samplei}}] = \frac{\sz_{i_0}}{\sum_{i}^k \sz_i}.\]
Next observe that the second term in \equref{appunionone}
is $\probunif{\sigma_0 \mbox{ is sampled in round $j_0$ from } S_{i_0}}{\accurateSmpevent{1,\ldots, k}}$.
This term, in turn, can be computed
by considering the disjoint union of the events corresponding to how many times
$i_0$ has been chosen so far: $\probunif{\sigma_0 \mbox{ is sampled in round $j_0$ from } S_{i_0}}{\accurateSmpevent{1,\ldots, k}}$
\begin{align*}
= &\sum\limits_{m=1}^t \prob{ \begin{aligned} i_0 \text{ was sampled $m-1$ times before round } j_0 \\
\text{ and } \sigma_0 \mbox{ is sampled in round $j_0$ from } S_{i_0} \end{aligned}; \accurateSmpevent{1,\ldots, k}} 
=  \sum\limits_{m=1}^t \probunif{S_{i_0}[m] = \sigma_0}{\accurateSmpevent{1,\ldots, k}} \\
= & \sum\limits_{\omega} \sum\limits_{m=1}^t \prob{\omega[m] = \sigma_0 \mid S_{i_0} = \omega} \cdot \prob{U_{i_0} = \omega} 
=  \sum\limits_{\omega} \prob{U_{i_0} = \omega \text{ and } \omega[m] = \sigma_0} = \frac{1}{|T_{i_0}|}
\end{align*}
Therefore, 
$\frac{1}{(1 + \epsilon_{\sz})\sum_{i=1}^k \sz_i} \leq  \probunif{(\sigma_0, i_0)  \mbox{ is sampled in round }j_0}{\accurateSmpevent{1,\ldots, k}} \leq (1 + \epsilon_{\sz})\frac{1}{\sum_{i=1}^k \sz_i}.$

For any $i_0 \in [t]$ and $\sigma_0\in S_{i_0}$, we have $ \probunif{(\sigma_0, i_0) \mbox{ is sampled }}{\accurateSmpevent{1,\ldots, k}} $
\begin{align*}
 = &  \probunif{i_0  \mbox{ is sampled in \lineref{ev-samplei}} }{\accurateSmpevent{1,\ldots, k}} \times \probunif{\sigma \mbox{ is obtained in \lineref{ev-samplesigma}} \mid i_0 \mbox{ is sampled in \lineref{ev-samplei}} }{\accurateSmpevent{1,\ldots, k}} \\
 = & \frac{\sz_{i_0}}{\sum_{i}^k \sz_i}\times \frac{1}{|T_{i_0}|}
\end{align*} 

Probability that $i_0$ is sampled in \lineref{ev-samplei} is independent of the distributions of the samples.
The last equality is because we have assumed $\forall i$, $|S_{i}| \geq t$, 
so the \textbf{if} condition in \lineref{ev-if} is always satisfied. 
Hence once $i_0$ is sampled in \lineref{ev-samplesigma} probability that 
$\sigma_0$ is obtained is exactly $1/|T_{i_0}|$, since the set $S_{i_0}$ 
contains elements sampled uniformly (with replacement) from $T_{i_0}$. Therefore, 
$\frac{1}{(1 + \epsilon_{\sz})\sum_{i=1}^k \sz_i} \leq  \probunif{\sigma_0, i_0)  \mbox{ is sampled }}{\accurateSmpevent{1,\ldots, k}} \leq (1 + \epsilon_{\sz})\frac{1}{\sum_{i=1}^k \sz_i}.$
 
 Let $Y_j$ be the random variable denoting whether the counter $Y$ 
 is increased in the $j$th iteration of the loop. 
 Note that $Y_j = 1$ if the pair $(\sigma_{0}, i_0)$ sampled in the $j$th 
 iteration is in $\mathcal{U}_{\textsf{unique}}$.   
 This is because in \lineref{ev-check} it is checked if the sampled pair is in $\mathcal{U}_{\textsf{unique}}$. 
 
 So, $\mathbb{E}[Y_j; \accurateSmpevent{1,\ldots, k}] = \sum_{(\sigma, i)\in \mathcal{U}_{\textsf{unique}}}\probunif{(\sigma, i) \mbox{ is sampled in round }j}{\accurateSmpevent{1,\ldots, k}}$; here the `;' in the expectation has an analogous definition to the case when used for probability.
 Thus,  
 $$ \frac{|\mathcal{U}_{\textsf{unique}}|}{(1 + \epsilon_{\sz}) \sum_{i=1}^k \sz_i} 
 \leq 
 \mathbb{E}[Y_j; \accurateSmpevent{1,\ldots, k}] 
 \leq 
 (1 + \epsilon_{\sz})  \frac{|\mathcal{U}_{\textsf{unique}}|}{\sum_{i=1}^k \sz_i}$$
 
 Note that $Y = \sum_{i}^t Y_i$. Thus at the end of the algorithm the expected value of $Y$ is between 
 $\frac{t |\mathcal{U}_{\textsf{unique}}|}{(1 + \epsilon_{\sz}) \sum_{i=1}^k \sz_i}$ and $(1 + \epsilon_{\sz})^{2}  \frac{t |\mathcal{U}_{\textsf{unique}}|}{\sum_{i=1}^k \sz_i}$.
 Now by Chernoff bound we have, 
 $$\probunif{Y \not\in (1 \pm \frac{\epsilon}{2})\mathbb{E}[Y_j; \accurateSmpevent{1,\ldots, k}]}{\accurateSmpevent{1,\ldots, k}} \leq 2exp(-\frac{\epsilon^2\mathbb{E}[Y_j; \accurateSmpevent{1,\ldots, k}]}{12}).$$
 
 We also know that,
  \begin{align*}
  \mathbb{E}[Y; \accurateSmpevent{1,\ldots, k}] \geq & \frac{t |\mathcal{U}_{\textsf{unique}}|}{(1 + \epsilon_{\sz}) \sum_{i=1}^k \sz_i} 
   \geq \frac{t\max_i |T_i|}{(1 + \epsilon_{\sz}) \sum_{i=1}^k \sz_i}
   \geq  \frac{t\max_i \sz_i}{(1 + \epsilon_{\sz})^2\sum_{i=1}^k \sz_i}
   \geq  12\log\left(\frac{4}{\delta}\right)
  \end{align*}
The last inequality is due to the setting of the parameter $t$. Thus,  by Chernoff bound, probability that 
 the output of the algorithm is not between $(1-\frac{\epsilon}{2})\frac{|\mathcal{U}_{\textsf{unique}}|}{(1+\epsilon_{\sz})}$ and $(1 + \frac{\epsilon}{2})(1+\epsilon_{\sz})|\mathcal{U}_{\textsf{unique}}|$ is at most $\delta/2$. 
 Since $(1-\frac{\epsilon}{2}) \geq \frac{1}{1 + \epsilon}$, so we have $\probunif{\mathsf{Fail}}{\accurateSmpevent{1,\ldots, k}} \leq \frac{\delta}{2}$.  \\

\textbf{Part 2: If for all $i$,  $|S_i|\geq \thresh$ then $\probunif{\mbox{ \lineref{ev-break} is ever reached }}{\accurateSmpevent{1,\ldots, k}}\leq \frac{\delta}{2}.$}

We first observe that since the probability of reaching \lineref{ev-break} is independent of the distributions of the sample, we have that $\probunif{\mbox{ \lineref{ev-break} is ever reached }}{\accurateSmpevent{1,\ldots, k}} = \Pr[\mbox{ \lineref{ev-break} is ever reached }]$, and thus we will try to upper bound the latter quantity instead.
Note that, the \textbf{else} clause in \lineref{ev-break} is ever reached, if for some $i$, the number of times $i$ is sampled in the $t$ iterations is more than \thresh. 

Let $\mathsf{Bad}_i$ denote the event that the index $i$ is sampled (in \lineref{ev-samplei}) more that $\thresh$ number of times during the course of the algorithm, that is during the $t$ runs of the loop 
(\linerangeref{ev-loopstart}{ev-loopend}). And let the event $\mathsf{Bad}$ be $\cup_{i=1}^k \mathsf{Bad}_i$. We will now upper bound, $\Pr[\mathsf{Bad}_i]$. 
Let $W_i$ be the random variable counting the number times $i$ is sampled: $\mathbb{E}[W_i] = t\times \frac{\sz_i}{\sum_{i=1}^k \sz_i}$ which is less than $\frac{12(1 + \epsilon_{\sz})^2}{\epsilon^2}\log(4/\delta)$.  
By simple application of Chernoff Bound we see that $$\Pr[W_i \geq (1+ \Delta)\mathbb{E}[W_i]] \leq exp\left(-\frac{\Delta^2 \mathbb{E}[W_i]}{2 + \Delta}\right) \leq exp\left(-\frac{\Delta\mathbb{E}[W_i]}{2}\right),$$
the last inequality is for $\Delta \geq 2$. For our purpose if we put $\Delta = \log(\frac{2k}{\delta})\frac{1}{\mathbb{E}[W_i]}$, since $|S_i| \geq \thresh \geq \frac{12\cdot (1+\epsl_{\sz})^2}{\varepsilon^2}\cdot \log(\frac{4}{\delta})$  + $\log(\frac{2k}{\delta})$, we have 
$\Pr[\mathsf{Bad_i}] \leq \frac{\delta}{2k}$. So by union bound $\Pr[\mathsf{Bad} ] \leq \delta/2$. 
Observe that for large $k$, $\Delta \geq 2$ as desired.
$\Pr[\mathsf{Bad_i}] \leq \frac{\delta}{2k}$.Observe that for large $k$, $\Delta \geq 2$ as desired. 
So by union bound, we have: 
$\probunif{\mathsf{Bad}}{\accurateSmpevent{1,\ldots, k}} = \Pr[\mathsf{Bad} ] \leq \delta/2.$
\end{proof}


\subsection{Correctness of~\algoref{sampling}}
\seclabel{samplingcorrectness}

We prove the three parts of \thmref{correctness-sampling-subroutine}
individually.

\subsection*{Proof of~\thmref{correctness-sampling-subroutine}\itmref{uniform-on-support}}

Consider an execution of~\algoref{sampling} and let 
$w = w_1 w_2 \ldots w_\lvl \in \lang{\stlvl{\stateq}{\lvl}}$ be the string 
constructed right before entering the branch at \lineref{overflow}.
Let $p_{\zz, i}$ be the value of the variable $\mathsf{pr_{\zz}}$
at level $i$ and let $v$ be the value of $\varphi$ before the function returns.
Then we have,
\begin{align*}
v = \frac{\gamma_0}{p_{w_\lvl, \lvl} \cdot p_{w_{\lvl-1}, \lvl-1} \cdots p_{w_1, 1}}
\end{align*}

Let $\sz^\lvl_{\zz}, \sz^\lvl_{\oo}, \sz^{\lvl-1}_{\zz},  \sz^{\lvl-1}_{\oo}, \ldots \sz^1_{\zz}, \sz^1_{\oo}$
be the estimates obtained in the $2\lvl$ calls to $\appdel$ during the run of $\SampleFun()$.
Then,
\begin{align*}
p_{w_i, i} = \frac{\sz^i_{w_i}}{\sz^i_{w_i} + \sz^i_{1-w_i}}
\end{align*}
Thus,
$
v = \frac{\gamma_0}{\prod\limits_{i=1}^\lvl \frac{\sz^i_{w_i}}{\sz^i_{w_i} + \sz^i_{1 - w_i}}}
 = \gamma_0 \cdot \prod\limits_{i=1}^\lvl \frac{\sz^i_{w_i} + \sz^i_{1 - w_i}}{\sz^i_{w_i}}
$

\noindent
Recall that $P^i$ is the argument of the call at level $i$,
and $P^i_{\bb}$ is the $\bb$-predecessor of $P^i$.
Observe that
$\lang{P^i} = \lang{P^i_{\zz}} \uplus \lang{P^i_{\oo}}$,
and thus, $
|\lang{P^i_{\zz}}| + |\lang{P^i_{\oo}}| = |\lang{P^i}|
$

\noindent

Hence, under the event $\accurateEVevent_{i, \zz} \cap \accurateEVevent_{i, \oo}$ we have
\begin{align}
\frac{\sz^i_{w_i} + \sz^i_{1 - w_i}}{\sz^i_{w_i}} 
\leq 
(1+\beta)^{2i} \cdot \frac{|\lang{P^{i}}|}{|\lang{P^{i-1}}|}
\end{align}

\noindent
This gives us,
\begin{align*}
\begin{array}{c}
v \leq 
\gamma_0 \cdot \prod\limits_{i=1}^\lvl (1+\beta)^{2i} \cdot \frac{|\lang{P^{i}}|}{|\lang{P^{i-1}}|} 
= \gamma_0{\cdot}\Bigg(\prod\limits_{i=1}^\lvl (1+\beta)^{2i}\Bigg){\cdot}\Bigg(\prod\limits_{i=1}^\lvl \frac{|\lang{P^{i}}|}{|\lang{P^{i-1}}|}\Bigg)
\end{array}
\end{align*}

\noindent
Now, observe that $|\lang{P^{0}}| = 1$ and $P^{\lvl} = \set{\stlvl{\stateq}{\lvl}}$. 

Thus, under the event
$\big( \bigcap_{\begin{subarray}{l}1 \leq i\leq \lvl\\ \bb \in \set{\zz, \oo} \end{subarray}} \accurateEVevent_{i, \bb} \big) = \accurateEVevent$, we have
\begin{align}
\begin{array}{c}
v \leq
\gamma_0{\cdot}\Bigg(\prod\limits_{i=1}^\lvl (1+\beta)^{2i}\Bigg){\cdot}|\lang{\stlvl{\stateq}{\lvl}}|
=
\gamma_0{\cdot} (1+\beta)^{\lvl(\lvl+1)} {\cdot}|\lang{\stlvl{\stateq}{\lvl}}|
\end{array}
\end{align}

\noindent
Assuming the event $\accurateEstevent_{\lvl, \stateq}$ also holds, we have 
$\frac{|\lang{\stlvl{\stateq}{\lvl}}|}{\approxcnt{\stlvl{\stateq}{\lvl}}} \leq (1+\beta)^{\lvl}$.
Under this assumption, substituting $\gamma_0 = \czero \cdot \frac{1}{\approxcnt{\stlvl{\stateq}{\lvl}}}$,
$\beta \leq \frac{1}{4\cdot n^2}$, and using the inequality $(1+1/4n^2)^n \leq e$, we have,
\begin{align*}
\begin{array}{c}
v \leq \czero{\cdot}(1+\beta)^{\lvl(\lvl+2)} \leq \czero \cdot e \leq 1 \\
\end{array}
\end{align*}

It then follows that $\prob{\failevent_1 \mid \accurateEVevent \cap \accurateEstevent_{\lvl, \stateq}}$ is equal to $0$..
Substituting, $v = \frac{\gamma_0}{p_{w_\lvl, \lvl} \cdot p_{w_{\lvl-1}, \lvl-1} \cdots p_{w_1, 1}} = \frac{\gamma_0}{ \prod_{i=\lvl}^{1} p_{\zz, i} }$, we then have
\begin{align*}
\prob{\rtrnsmp = w \mid \accurateEVevent \cap \accurateEstevent_{\lvl, \stateq}}
= v \cdot \prod_{i=\lvl}^{1} p_{\zz, i}
= \frac{\gamma_0}{ \prod_{i=\lvl}^{1} p_{\zz, i} } \cdot \prod_{i=\lvl}^{1} p_{\zz, i} = \gamma_0
\end{align*}

\subsection*{Proof of~\thmref{correctness-sampling-subroutine}\itmref{failure-given-accurate}}

We will now estimate the (conditional) probability of the event $\failevent$.
First, observe that for each string $w' \notin \lang{\stlvl{\stateq}{\lvl}}$, 
we have $\prob{\rtrnsmp=w'} = 0$ because  any string $w$ that the algorithm outputs is such that
 there is a path from the initial state to $\stlvl{\stateq}{\lvl}$
 labeled with $w$, and 
 (b) the unrolled automata is such that any path from start state to
 some state $\stlvl{\statep}{i}$ corresponds to a string in $\lang{\stlvl{\statep}{i}}$.

We observe that $\overline{\failevent}$ is the event that
$\rtrnsmp$ is a string in $\lang{\stlvl{\stateq}{\lvl}}$, each of which is a disjoint event.
Thus,
\begin{align*}
&\prob{\failevent \mid \accurateEVevent \cap \accurateEstevent_{\lvl, \stateq}} \\
 &= 1{-}\!\!\!\!\!\!\sum\limits_{w \in \lang{\stlvl{\stateq}{\lvl}}}\!\!\!\prob{\rtrnsmp = w \mid \accurateEVevent \cap \accurateEstevent_{\lvl, \stateq}}\\
 &= 1{-}\!\!\!\!\!\!\sum\limits_{w \in \lang{\stlvl{\stateq}{\lvl}}}\!\!\!\gamma_0 =  1 - \gamma_0|\lang{\stlvl{\stateq}{\lvl}}| \leq 1 - \czero \cdot \frac{|\lang{\stlvl{\stateq}{\lvl}}|}{\approxcnt{\stlvl{\stateq}{\lvl}}}  \leq 1 - \czero \cdot \frac{1}{(1+\beta)^{\lvl}} \leq 1 - \frac{2}{3e^2}
\end{align*}

\subsection*{Proof of~\thmref{correctness-sampling-subroutine}\itmref{accurate-given-good}}

We would like to use \thmref{approx-delphic-correct}
to prove \thmref{correctness-sampling-subroutine}\itmref{accurate-given-good}.
Towards this, we first establish that the pre-conditions of~\thmref{approx-delphic-correct} hold.

\begin{enumerate}
    \item For $i \leq \lvl$, if we substitute $\beta' = (1+\beta)^{i-1}-1, \eta'=\eta/{4n}$, we get 
    $\thresh = 24 \cdot \frac{(1+\beta)^{2i-2}}{\beta^2} \cdot \log \left( \frac{16k_{\bb}n}{\eta} \right)$,
    where $k_{\bb}$ is the number of sets in the argument of $\appdel$.
    Observe that, $k_{\bb} \leq m$ and that,
    for each $i \leq \lvl$,
    we have $(1+\beta)^{2i-2} \leq (1+\beta)^{4n^2} \leq e$.
    Thus for each $\stlvl{\stateq}{i}$, we have
    \[
    \thresh \leq \frac{24 \cdot e}{\beta^2} \log \left(\frac{16mn}{\eta} \right) \leq |\sample{\stlvl{\stateq}{i-1}}|.
    \]

    \item Since we are considering the distribution
    obtained when $\sample{\stlvl{\stateq}{i}}$
    is replaced with $\usample{\stlvl{\stateq}{i}}$
    for each $i \leq \lvl-1$
    (observe the $\probunif{}{\accurateSmpevent{\leq \lvl-1}}$ in the statement of 
    \thmref{correctness-sampling-subroutine}\itmref{accurate-given-good}),
    we directly have that each input $\set{S_i}_i$ to every call to
    $\appdel$ has 
   at least $\thresh$ elements (as established)
    elements obtained by repeatedly
    sampling $\set{T_i}_{i}$ uniformly.

    \item The event $\accurateEstevent_{\leq \lvl-1}$ demands that
    for each $i \leq \lvl-1$, we have
    $$\frac{|\lang{\stlvl{\stateq}{i}}|}{(1+\beta)^{i}} \leq \approxcnt{\stlvl{\stateq}{i}}  \leq (1+\beta)^{i} |\lang{\stlvl{\stateq}{i}}|.$$
\end{enumerate}

From \thmref{approx-delphic-correct}, it thus follows that, 
for all $i$, 
\begin{align*}
\probunif{  \accurateEVevent_{i, \bb}  \; \; \Bigg\vert \;\  
\accurateEstevent_{\leq \lvl-1}
}{\accurateSmpevent{\leq \lvl-1}}
\geq  1- \eta'
\end{align*}

Therefore, applying union-bound, we have 
\begin{align*}
\begin{array}{l}
\probunif{ \bigcap_{\begin{subarray}{l}1 \leq i\leq \lvl\\ \bb \in \set{\zz, \oo} \end{subarray}} \accurateEVevent_{i, \bb}  \; \; \Bigg\vert \;\  
\accurateEstevent_{\leq \lvl-1}
}{\accurateSmpevent{\leq \lvl-1}} 
\geq  1- 2n\cdot\eta' = 1-\frac{\eta}{2}
\end{array}
\end{align*}







\subsection{Correctness of~\algoref{main}}
\seclabel{mainalgo}

We prove the main result (\thmref{main}) by induction on the level $\lvl$.
For each level $\lvl$, we first characterize the accuracy of the computed estimates
$\set{\approxcnt{\stlvl{\stateq}{\lvl}}}_{\stateq \in \states}$ 
and the quality of the sampled sets $\set{\sample{\stlvl{\stateq}{\lvl}}}_{\stateq \in \states}$,
assuming that these were computed when the samples
at the lower levels $1, 2, \ldots, \lvl-1$ are perfectly uniform and have size $\thresh$.
After this, we characterize the real computed estimates and samples by arguing
that the distance of the corresponding random variables, from those computed using
perfectly uniform samples, is small.
We first establish some helpful auxiliary results (\lemref{iterative-count-samp} 
and \lemref{sampbound-iterative}) towards our main proof.
As before, 
we will often invoke distribution entanglement of
${\sf S}_{\leq \lvl-1}$  
by $\accurateSmpevent{\leq \lvl-1}$ for each level $\lvl$.

\begin{lemma}
\lemlabel{iterative-count-samp}
For each $\lvl \leq n$ and $\stateq \in \states$, we have
\begin{align*}
    \probunif{ 
    \accurateEstevent_{\stateq, \lvl}
 \;\; \big\vert \;\;     \accurateEstevent_{\leq \lvl-1}       
    }{\accurateSmpevent{\leq \lvl-1}} \geq 1-\eta
\end{align*}
\end{lemma}

\begin{proof}
Let us denote by
$\lang{\pred{\stlvl{\stateq}{\lvl}}{\bb}}$
the set
$\bigcup\limits_{\statep \in \pred{\stlvl{\stateq}{\lvl}}{\bb}} \lang{p}.$
Observe that,
$$
\lang{\stlvl{\stateq}{\lvl}} = 
\lang{\pred{\stlvl{\stateq}{\lvl}}{\zz}} \cdot \set{\zz} \uplus \lang{\pred{\stlvl{\stateq}{\lvl}}{\oo}} \cdot \set{\oo}.
$$
Now, we would like to apply \thmref{approx-delphic-correct} 
to show that $\sz_{\bb}$ (at level $\lvl$) is a {\em good} estimate. 
To this end, we first show that the pre-conditions stated in \thmref{approx-delphic-correct} hold.
This is because, 
(1) under the event
 $\accurateEstevent_{\leq \lvl-1}$,
the estimates
$\set{\approxcnt{\stlvl{\statep}{\lvl-1}}}_{\statep \in \states}$ satisfy desired requirements. 
Next,
(2) 
for each $i \leq \lvl-1$, the samples
 $\set{\sample{\stlvl{\statep}{\lvl-1}}}_{\statep \in \states}$
 also obey the desired requirements.
 This is because $\thresh$
 (in \algoref{approx-delphic}), after substituting $\beta, \eta$, satisfies 
\begin{align*}
    \thresh &= 24   \cdot \frac{(1+\beta)^{2\lvl-2}}{\beta^2} \cdot \log \left(\frac{4k_{\bb}}{\eta'}\right)
    \leq \frac{24 \cdot e}{\beta^2} \log \left(\frac{8m^2n}{\delta}\right) \leq \frac{384 \cdot e \cdot n^4}{\epsl^2} \log \left(\frac{8m^2}{\delta} \right) \leq \NumSamples 
\end{align*}
Therefore, from \thmref{approx-delphic-correct}, we have 
\begin{align*}
    \probunif{
    \frac{|L(\pred{\stateq^{\lvl}}{\bb})|}{(1+\beta)^\lvl} 
     \leq \sz_{\bb}  
     \leq (1+\beta)^\lvl \; |L(\pred{\stateq^{\lvl}}{\bb})| 
	\; \big\vert \;  
    \accurateEstevent_{\leq \lvl-1}
    }
    {\accurateSmpevent{\leq \lvl-1}}
\end{align*}

 is at least $(1-\frac{\eta}{2})$.    Since $\approxcnt{\stlvl{\stateq}{\lvl}} = \sz_{\zz} + \sz_{\oo}$, we have 
\begin{align*}
    \probunif{
     \frac{|\lang{\stlvl{\stateq}{\lvl}}|}{(1+\beta)^\lvl} \leq \approxcnt{\stlvl{\stateq}{\lvl}}  \leq (1+\beta)^\lvl |\lang{\stlvl{\stateq}{\lvl}}|
    \; \big\vert \; \accurateEstevent_{\leq \lvl-1}
    }
    {\accurateSmpevent{\leq \lvl-1}} \geq 1-\eta \qquad \qedhere
\end{align*}
 \end{proof}

In addition to the previously introduced notation $\probunif{}{\accurateSmpevent{\leq i}}$,
we will also use the notation $\TVDistUnif{X, Y}{\accurateSmpevent{\leq i}}$
(resp. $\TVDistUnif{X, Y \, \big\vert \,  \mathcal{F}}{\accurateSmpevent{\leq i}}$)
to denote the total variational distance between the distribution
induced by random variables $X, Y$ (resp. 
distribution induced by $X, Y$ conditioned on the event $\mathcal{F}$) 
when ${\sf S}_{\leq i}$ is distribution entangled to $\accurateSmpevent{\leq i}$.
Formally,
\begin{align*}
	\TVDistUnif{X, Y \mid \mathcal{F}}{\accurateSmpevent{\leq i}}
	 = \sum_{\omega'} &\prob{\accurateSmpevent{\leq i} = \omega'} 
	\left(\sum_{\omega}
	\Big(\Pr[X = \omega \mid \mathcal{F} \cap ({\sf S}_{\leq i} = \omega')] \right. \\
	& - \min \left(\Pr[X = \omega \mid \mathcal{F} \cap ({\sf S}_{\leq i} = \omega')] , 
	\Pr[Y = \omega \mid \mathcal{F} \cap ({\sf S}_{\leq i} = \omega')]\right) \Big) \Big)
\end{align*}
In the above expression,
$\omega$  ranges over the union of domains of $X$ and $Y$ while $\omega'$ ranges 
over the union of domains of ${\sf S}_{\leq i}$ and $\accurateSmpevent{\leq i}$.

\begin{lemma}
\lemlabel{sampbound-iterative}
For each $\lvl \leq n$ and $\stateq \in \states$, we have
\begin{align*}
\TVDistUnif{
\sample{\stlvl{\stateq}{\lvl}}, \usample{\stlvl{\stateq}{\lvl}}) \;\; \big\vert \;\;  \accurateEstevent_{\leq \lvl}}
{\accurateSmpevent{\leq \lvl-1}}
 \leq \eta 
\end{align*}    
\end{lemma}

\begin{proof}
Fix $\lvl$ and $\stateq$. 
Let $\failevent^j$ denote the event that the $j^\text{th}$ call to {\SampleFun} (in the loop starting at \lineref{line7}) for $i=j$ returns $\bot$. 
Similarly, let  $\accurateEVevent^j$  denote the event $\accurateEVevent$ corresponding to the call to {\SampleFun} for $i=j$.
Using Iverson bracket notation~\cite{Ive62}, we denote the indicator variable $[\failevent^j]$. Also, we use $\smallset$ to denote the event that 
$\alpha_{\stlvl{\stateq}{\lvl}} > 0$.
Observe that  $\smallset$ occurs if and only if  $\sum\limits_{j=1}^{\NumSamplesExtra}[\failevent^j] \geq \NumSamplesExtra-\NumSamples$. 

Recall, we have, for each $j$,   
$\probunif{ [\failevent^j] \;\; \big\vert \;\;  \accurateEVevent^{j} \cap \accurateEstevent_{\leq \lvl}
}{\accurateSmpevent{\leq \lvl-1}} $
is at most  $(1-\frac{2}{3e^2})$.
Therefore, from Chernoff bound, we have 
\begin{align*}
    \probunif{ \smallset \;\; \big\vert \;\;  \bigcap\limits_{j} \accurateEVevent^j
    \cap \accurateEstevent_{\leq \lvl}    
    }{\accurateSmpevent{\leq \lvl-1}} \leq \frac{\eta}{2}.
\end{align*}
Furthermore, from~\thmref{correctness-sampling-subroutine} (part 3), we have 
\begin{align*}
\probunif{\overline{\bigcap\limits_{j} \accurateEVevent^j} \big\vert \;\; 
\accurateEstevent_{\leq \lvl}      
}{\accurateSmpevent{\leq \lvl-1}} \leq \frac{\eta}{2}.
\end{align*}

Note that given the event $\accurateEstevent_{\leq \lvl}$ and 
under the distributions obtained using $\accurateSmpevent{\leq \lvl-1}$, 
if the event $\smallset$ does not happen and the event $\bigcap_{j}\accurateEVevent^j$ holds then the distribution of $\sample{q^{\lvl}}$ and $\usample{q^{\lvl}}$ is identical. So, 
\begin{align*}
    &\TVDist{S(q^{\lvl}), U(q^{\lvl}) \;\; \big\vert \;\; 
    \accurateEstevent_{\leq \lvl}\; ;\;      \accurateSmpevent{\leq \lvl-1}}\\ 
    & \leq 
    \probunif{\smallset \cup \overline{\bigcap\limits_{j} \accurateEVevent^j} \;\; \big\vert \;\;  \accurateEstevent_{\leq \lvl}}
    {\accurateSmpevent{\leq \lvl-1}}
    \\
   & \leq 
   \probunif{\smallset \;\; \big\vert \bigcap\limits_{j} \accurateEVevent^j \cap \accurateEstevent_{\leq \lvl}}
   {\accurateSmpevent{\leq \lvl-1}} 
   \\  
& \quad \qquad \qquad + \probunif{\overline{\bigcap\limits_{j} \accurateEVevent^j} \;\; \big\vert \;\; 
\accurateEstevent_{\leq \lvl}  
   }{\accurateSmpevent{\leq \lvl-1}} \leq \eta/2 + \eta/2 = \eta \qedhere
\end{align*}
\end{proof}

\vspace{0.1in}
\subsubsection{Proof of \thmref{main}}

Our proof of the \thmref{main} relies on upper bounding
the probability with which the estimate of \algoref{main} is outside the desired range.
For this, we will set up new random variables and relate
its joint distribution with the statsitical distance between 
the samples constructed by the algorithm and the ideal set of uniform samples. 
Recall that  $\set{\sample{\stlvl{\stateq}{\lvl}}}_{\lvl \in [0,n], \stateq \in \states}$ 
represent  the set of random variables corresponding to the 
(multi)sets of  samples  we obtain during run of~\algoref{main}. 
Furthermore,  $\set{\usample{\stlvl{\stateq}{\lvl}}}_{\lvl \in [0,n], \stateq \in \states}$ 
represent  the set of random variables that correspond to  repeatedly 
sampling $\NumSamples$ elements (with replacement) 
uniformly at random from $\lang{\stlvl{\stateq}{\lvl}}$.
It is worth observing that while, 
for all $\stateq, \statep, \lvl, \lvl'$, 
$\sample{\stlvl{\stateq}{\lvl}}$ and $\sample{\stlvl{\statep}{\lvl'}}$ 
are dependent but $\usample{\stlvl{\stateq}{\lvl}}$ and 
$\usample{\stlvl{\stateq}{\lvl'}}$ are independent. 
Also, the definition of distribution $\usample{\stlvl{\stateq}{\lvl}}$ 
is independent of ~\algoref{main} as it is defined solely based 
on the set of $\lang{\sample{\stlvl{\stateq}{\lvl}}}$. 
 Since we are talking about the sets of random variables, 
 there is an implicit assumption that these sets are ordered.

We will now define two sequences of random variables:  
$$X = \set{X^{\stateq}_{\lvl}}_{\lvl \in [0,n], \stateq \in \states} \ \mbox{ and }\ 
Y =  \set{Y^{\stateq}_{\lvl}}_{\lvl \in [0,n], \stateq \in \states},$$
Instead of explicitly specifying them,
we will instead specify properties about their joint distribution.  
We will use the notation
$X_{\leq i}$ (and resp. $Y_{\leq i})$) to denote
the ordered set $\set{X^{\stateq}_{\lvl}}_{\lvl \in [1,i], \stateq \in \states}$ 
(resp. the ordered set $\set{Y^{\stateq}_{\lvl}}_{\lvl \in [1,i], \stateq \in \states}$).  
The joint distribution of $(X,Y)$ has the following properties:
\begin{enumerate}
	\item For every $\stateq \in \states$, 
		$X^{\stateq}_{0} = Y^{\stateq}_{0} = {\sf S}^{\stateq}_{0}$.
	\item For every $\stateq \in \states, \lvl \in [1, n]$,
		$X^{\stateq}_{\lvl} = S(\stateq^{\lvl})$.
	\item \label{cond:joint} 
	For all sets
	$\omega$, sequences of sets $\omega'$, 
	$\lvl \in [1, n]$, and $\stateq \in \states$,  we have
	\begin{align*}
		&\prob{X^{\stateq}_{\lvl} = Y^{\stateq}_{\lvl} = \omega 
			\mid 
			(X_{\leq \lvl-1}= Y_{\leq \lvl-1} = \omega') \cap \accurateEstevent_{\leq \lvl}}  \\ 
		&\qquad \qquad = \min \left(
		 \prob{U(q^{\ell}) = \omega},    \prob{X_{\lvl}^{\stateq} = \omega\mid (X_{\leq \ell-1}=  \omega') \cap \accurateEstevent_{\leq \ell}} \right) \\
		& \prob{X_{\lvl}^{\stateq} = \omega \mid (X_{\leq \lvl-1}= Y_{\leq \lvl-1} = \omega') \cap \accurateEstevent_{\leq \ell}}   \leq  \prob{X_{\lvl}^{\stateq} = \omega\mid (X_{\leq \ell-1}=  \omega') \cap \accurateEstevent_{\leq \ell}}
	\end{align*}
\end{enumerate}


We remark that such a joint distribution always exists.
This is because (\ref{cond:joint}) only restricts the case when 
$X_{\lvl}^{\stateq} = Y_{\lvl}^{\stateq}$,
thereby, providing enough choices for assignments of probability 
values to the cases where $X_{\lvl}^{\stateq} \neq Y_{\lvl}^{\stateq}$ 
such that $Y^{\stateq}_{\lvl}$ has a well-defined probability distribution. 
In \appref{joint-distribution} we give a concrete distribution that realizes
these properties.

\noindent
We will begin with the following simple observation, whose proof is deferred to Appendix.
\begin{restatable}{claim}{xboundclaim}
		\label{claim:xbound}
$\prob{(X_{\leq \lvl-1} =Y_{\leq \lvl-1} = \omega) \cap \accurateEstevent_{\leq \lvl-1}} 
\leq  
\prob{\accurateSmpevent{\leq \lvl-1} = \omega}$ 
\end{restatable}

Next, we relate the random variables defined above to the desired statistical distance:

\begin{claim}
\claimlabel{claim-1-main-proof}
	\begin{align*}
	 \prob{(X_{\lvl}^{\stateq} \neq Y_{\lvl}^{\stateq}) \cap (X_{\leq \lvl-1} = Y_{\leq \lvl-1})   \cap   \accurateEstevent_{\leq \lvl}} 
	& \leq \TVDistUnif{\sample{\stlvl{\stateq}{\lvl}}, \usample{\stlvl{\stateq}{\lvl}} \; \big\vert \;  \accurateEstevent_{\leq \lvl}}{\accurateSmpevent{\leq \lvl-1}}
	\end{align*}
\end{claim}
\begin{proof} $\prob{(X_{\ell}^{q} \neq Y_{\ell}^{q}) \cap (X_{\leq \ell-1} = Y_{\leq \ell-1})   \cap  \accurateEstevent_{\leq \lvl}}$
	
\begin{align*} 
	\label{eq:coupling}
&	= \sum_{\omega,\omega'}	\prob{(X_{\ell}^{q} = \omega) \cap (X_{\leq \ell-1} =Y_{\leq \ell-1} = \omega')   \cap   \accurateEstevent_{\leq \lvl}} \\
	   &\quad - \sum_{\omega,\omega'} \prob{(X_{\ell}^{q} = Y_{\ell}^{q} =  \omega) \cap (X_{\leq \ell-1} =   Y_{\leq \ell-1} = \omega')  \cap  \accurateEstevent_{\leq \lvl}}\\
	&= \sum_{\omega'} \prob{(X_{\leq \ell-1} =Y_{\leq \ell-1} = \omega') \cap \accurateEstevent_{\leq \lvl}} \cdot \\ 
	& \qquad \left(\sum_{\omega}	\prob{(X_{\ell}^{q} = \omega \mid (X_{\leq \ell-1} =Y_{\leq \ell-1} = \omega') \cap   \accurateEstevent_{\leq \lvl}}\right. \\
& \left. \qquad - \sum_{\omega} \prob{(X_{\ell}^{q} = Y_{\ell}^{q} =  \omega \mid (X_{\leq \ell-1} =   Y_{\leq \ell-1} = \omega')  \cap  \accurateEstevent_{\leq \lvl}} \right)\\
&\leq \sum_{\omega'} \prob{U_{\leq \ell-1} = \omega'} \cdot \big(   \sum_{\omega}	\left(\prob{(X_{\ell}^{q} = \omega \mid (X_{\leq \ell-1}  = \omega') \cap   \accurateEstevent_{\leq \lvl}} \right. \\
&  \qquad \qquad - \min \left(  \prob{U(q^{\ell}) = \omega},  \left.  \left. \prob{X_{\ell}^{q} = \omega \mid ( X_{\leq \ell-1} =   \omega')  \cap  \accurateEstevent_{\leq \lvl} }   \right) \right) \right)  \\
& = \TVDistUnif{\sample{\stlvl{\stateq}{\lvl}}, \usample{\stlvl{\stateq}{\lvl}} \; \big\vert \;  \accurateEstevent_{\leq \lvl}}{\accurateSmpevent{\leq \lvl-1}}
\end{align*}

\end{proof}

We now focus on deriving an upper bound on $\accurateEstevent_{n}$ as stated in the following claim, whose proof is  deferred to Appendix.
\begin{restatable}{claim}{accevent}
	\label{claim:fermat}
	
	$\overline{\accurateEstevent_{n}}   \subseteq 
		\bigcup_{\ell=1}^{n-1}(\overline{\accurateEstevent_{\ell}} \cap 
		{\accurateEstevent_{\leq \ell-1}} \cap
		(X_{\leq \ell-1} = Y_{\leq \ell-1}) ) $
	\begin{align*}
		& \qquad \qquad \qquad \qquad  \cup \bigcup_{\ell=1}^{n-1} \left(
		(X_{ \ell} \neq Y_{ \ell}) \cap
		(X_{\leq \ell-1} = Y_{\leq \ell-1}) \cap
		{\accurateEstevent_{\leq \ell}} 
		\right) 
	\end{align*}
\end{restatable}

\noindent
We can now finish the proof of \thmref{main}. 
We first observe that:
\begin{align*}
		& \prob{\overline{\accurateEstevent_{\ell}} \cap 
		{\accurateEstevent_{\leq \ell-1}} \cap
		(X_{\leq \ell-1} = Y_{\leq \ell-1})} \\
	&	\leq 
	  \sum_{\omega,q} \prob	{\accurateEstevent_{\leq \lvl-1} \cap
		(X_{\leq \ell-1} = Y_{\leq \ell-1} = \omega) }\times \\
	& \qquad\qquad\qquad\qquad		\prob{\overline{\accurateEstevent_{q,\ell}} \mid 
		(X_{\leq \ell-1} = Y_{\leq \ell-1}=\omega) \cap  
		{\accurateEstevent_{\leq \ell-1}} } \\
	& \leq \sum_{\omega,q} \prob	{U_{\leq \ell-1} = \omega}\prob{\overline{\accurateEstevent_{q,\ell}} \mid 
		(X_{\leq \ell-1} = Y_{\leq \ell-1} = \omega)
		 \cap  
		{\accurateEstevent_{\leq \ell-1}} } \\
	& =	 \sum_{\omega,q} \prob	{\accurateEstevent_{\leq \lvl-1} = \omega}\prob{\overline{\accurateEstevent_{q,\ell}} \mid 
	({\sf S}_{\leq \lvl-1} = \omega \cap Y_{\leq \lvl-1} = \omega
	\cap  
	{\accurateEstevent_{\leq \lvl-1}} } \\
	& \leq  \sum_{q}	\probunif{
		\overline{\accurateEstevent_{q,\lvl}} \;\; \big\vert \;\;
		\accurateEstevent_{\leq \lvl-1}
	}{
		\accurateSmpevent{\leq \lvl-1}
	}
\end{align*}

\noindent
The final step in the derivation uses the observation that given ${\sf S}_{\leq \lvl-1} $ 
and ${\accurateEstevent_{\leq \lvl-1}}$,  
the probability of the event $\accurateEstevent_{q,\lvl}$ 
is only dependent on the randomness of the algorithm {\appdel}. 
Thus, given ${\sf S}_{\leq \lvl-1} $ and 
${\accurateEstevent_{\leq \lvl-1}}$,  
the event $\accurateEstevent_{q,\lvl}$ is conditionally independent of $Y_{\leq \lvl-1}$. \\

\noindent
Putting together, we have:
\begin{align*}
	\prob{\overline{\accurateEstevent_{n}}} 
	\leq & \sum_{\lvl \in [1,n]; q \in Q} 
	\probunif{
		\overline{\accurateEstevent_{q,\lvl}} \;\; \big\vert \;\;
		\accurateEstevent_{\leq \lvl-1}
	}{
		\accurateSmpevent{\leq \lvl-1}
	} \\
	&\qquad 
	+ \sum\limits_{\lvl \in [1,n-1]; q \in Q} 
		\TVDistUnif{\sample{\stlvl{\stateq}{\lvl}}, \usample{\stlvl{\stateq}{\lvl}} 
		\mid 
		\accurateEstevent_{\leq \lvl}}{
		\accurateSmpevent{\leq \lvl-1}}
\end{align*}

By \lemref{iterative-count-samp} and \lemref{sampbound-iterative} and applying union bound over all $\lvl\in [1, n]$ and all $q \in Q$ we have 

$$\prob{\overline{\accurateEstevent_{n}}} 
\leq n\cdot m \cdot \eta+n\cdot m \cdot \eta \leq \delta .$$

To calculate the time complexity of~\algoref{main} 
we first observe that the \textbf{for} loop goes over $nm$ elements. 
Now for each $\lvl$ and $\stateq$ there is one call to 
$\appdel_{\beta, \eta}$ with the number of sets being $O(m)$ and 
$\NumSamplesExtra$ calls to $\SampleFun{}$. 
So by \thmref{approx-delphic-correct}, 
the time taken by  $\appdel$  
is the time it takes to do $O(m/\beta^2\log(m/\eta))$ number of 
calls to the membership oracle. 
So the number of calls to the membership oracle
is $\widetilde{O}(mn^4/\epsilon^2\log(1/\delta))$, 
ignoring the factors of $\log (n+m)$. 
The method $\SampleFun{}$ on the other hand is a recursive 
function which calls $\appdel_{\beta, \eta}()$ in every recursive step 
and the depth of the recursion is $O(n)$. 
So the time complexity for each of the $\NumSamplesExtra$ calls to 
$\SampleFun{}$ is the time it takes for
$\widetilde{O}(mn^5/\epsilon^2\log(1/\delta))$ calls to the oracle. 
{{\color{black}
Now, we note that the time complexity of the membership calls can be 
amortized as follows.
First, for every string 
$w \in \bigcup\limits_{\stateq, \lvl}\set{\sample{\stlvl{\stateq}{\lvl}}}$,
we pre-compute and store the set of reachable states of $w$ in $\aut$
in time $m^2|w| \in O(m^2n)$ time.
This takes a total of $\widetilde{O}(m^2n \cdot mn^5/\epsilon^2\log(1/\delta)) = \widetilde{O}(m^3n^6/\epsilon^2\log(1/\delta))$ time.
The membership checks (in \algoref{approx-delphic}, \lineref{ev-check})
can now be performed in $O(1)$ time.
So in total the time complexity of the~\algoref{main} is 
$\widetilde{O}(mn\cdot \NumSamplesExtra \cdot  \frac{mn^5}{\epsilon^2}\log(1/\delta) + m^3n^6/\epsilon^4\log(1/\delta))$ which is $\widetilde{O}((m^2n^{10} + m^3n^6)\cdot \frac{1}{\epsilon^4}\cdot\log^2(1/\delta))$.


\section{Conclusions and Future Work}
\seclabel{conclusions}

We consider the approximate counting problem
\#NFA and propose a fully polynomial time randomized
approximation scheme (FPRAS) that significantly improves
the prior FPRAS~\cite{ACJR19}.
Given the wide range of applications of counting and sampling
from the language of an NFA in applications
including databases, program analysis, testing and more broadly
in Computer Science, we envision that further improvements in the
complexity of approximating \#NFA is a worthwhile avenue for future work.

\begin{acks}
We want to express our sincerest gratitude to Weiming Feng for carefully reviewing the proofs and identifying several issues present in earlier versions of the paper. We are grateful  to Yash Pote, Uddalok Sarkar, and the anonymous reviewers for their constructive feedback, which enhanced the presentation and brought attention to relevant prior research.  Umang Mathur was partially supported by a Singapore Ministry of Education (MoE) Academic Research Fund (AcRF) Tier 1 grant. Kuldeep Meel's work was partly supported in part by the National Research Foundation Singapore [NRF-NRFFAI1-2019-0004], Ministry of Education Singapore Tier 2 grant MOE-T2EP20121-0011, and Ministry of Education Singapore Tier 1 Grant [R-252-000-B59-114].
\end{acks}

\bibliographystyle{ACM-Reference-Format}
\bibliography{references}


\begin{thebibliography}{18}


\ifx \showCODEN    \undefined \def \showCODEN     #1{\unskip}     \fi
\ifx \showDOI      \undefined \def \showDOI       #1{#1}\fi
\ifx \showISBNx    \undefined \def \showISBNx     #1{\unskip}     \fi
\ifx \showISBNxiii \undefined \def \showISBNxiii  #1{\unskip}     \fi
\ifx \showISSN     \undefined \def \showISSN      #1{\unskip}     \fi
\ifx \showLCCN     \undefined \def \showLCCN      #1{\unskip}     \fi
\ifx \shownote     \undefined \def \shownote      #1{#1}          \fi
\ifx \showarticletitle \undefined \def \showarticletitle #1{#1}   \fi
\ifx \showURL      \undefined \def \showURL       {\relax}        \fi
\providecommand\bibfield[2]{#2}
\providecommand\bibinfo[2]{#2}
\providecommand\natexlab[1]{#1}
\providecommand\showeprint[2][]{arXiv:#2}

\bibitem[Amarilli et~al\mbox{.}(2024)]%
        {AvBM23}
\bibfield{author}{\bibinfo{person}{Antoine Amarilli}, \bibinfo{person}{Timothy
  van Bremen}, {and} \bibinfo{person}{Kuldeep~S. Meel}.}
  \bibinfo{year}{2024}\natexlab{}.
\newblock \showarticletitle{Conjunctive Queries on Probabilistic Graphs: The
  Limits of Approximability}. In \bibinfo{booktitle}{\emph{27th International
  Conference on Database Theory, {ICDT}}}, Vol.~\bibinfo{volume}{290}.
  \bibinfo{pages}{15:1--15:20}.
\newblock
\urldef\tempurl%
\url{https://doi.org/10.4230/LIPICS.ICDT.2024.15}
\showDOI{\tempurl}


\bibitem[Angles et~al\mbox{.}(2017)]%
        {Angles2017}
\bibfield{author}{\bibinfo{person}{Renzo Angles}, \bibinfo{person}{Marcelo
  Arenas}, \bibinfo{person}{Pablo Barcel\'{o}}, \bibinfo{person}{Aidan Hogan},
  \bibinfo{person}{Juan Reutter}, {and} \bibinfo{person}{Domagoj Vrgo\v{c}}.}
  \bibinfo{year}{2017}\natexlab{}.
\newblock \showarticletitle{Foundations of Modern Query Languages for Graph
  Databases}.
\newblock \bibinfo{journal}{\emph{ACM Comput. Surv.}} \bibinfo{volume}{50},
  \bibinfo{number}{5}, Article \bibinfo{articleno}{68} (\bibinfo{year}{2017}).
\newblock
\urldef\tempurl%
\url{https://doi.org/10.1145/3104031}
\showDOI{\tempurl}


\bibitem[Arenas et~al\mbox{.}(2019)]%
        {ACJR19}
\bibfield{author}{\bibinfo{person}{Marcelo Arenas},
  \bibinfo{person}{Luis~Alberto Croquevielle}, \bibinfo{person}{Rajesh
  Jayaram}, {and} \bibinfo{person}{Cristian Riveros}.}
  \bibinfo{year}{2019}\natexlab{}.
\newblock \showarticletitle{Efficient Logspace Classes for Enumeration,
  Counting, and Uniform Generation}. In \bibinfo{booktitle}{\emph{Proceedings
  of the 38th ACM SIGMOD-SIGACT-SIGAI Symposium on Principles of Database
  Systems (PODS)}}. \bibinfo{pages}{59–73}.
\newblock
\urldef\tempurl%
\url{https://doi.org/10.1145/3294052.3319704}
\showDOI{\tempurl}


\bibitem[Arenas et~al\mbox{.}(2021)]%
        {ACJR21}
\bibfield{author}{\bibinfo{person}{Marcelo Arenas},
  \bibinfo{person}{Luis~Alberto Croquevielle}, \bibinfo{person}{Rajesh
  Jayaram}, {and} \bibinfo{person}{Cristian Riveros}.}
  \bibinfo{year}{2021}\natexlab{}.
\newblock \showarticletitle{\# NFA Admits an FPRAS: Efficient Enumeration,
  Counting, and Uniform Generation for Logspace Classes}.
\newblock \bibinfo{journal}{\emph{Journal of the ACM (JACM)}}
  \bibinfo{volume}{68}, \bibinfo{number}{6} (\bibinfo{year}{2021}),
  \bibinfo{pages}{1--40}.
\newblock


\bibitem[Bang et~al\mbox{.}(2016)]%
        {Bang2016}
\bibfield{author}{\bibinfo{person}{Lucas Bang}, \bibinfo{person}{Abdulbaki
  Aydin}, \bibinfo{person}{Quoc-Sang Phan}, \bibinfo{person}{Corina~S.
  P\u{a}s\u{a}reanu}, {and} \bibinfo{person}{Tevfik Bultan}.}
  \bibinfo{year}{2016}\natexlab{}.
\newblock \showarticletitle{String Analysis for Side Channels with Segmented
  Oracles}. In \bibinfo{booktitle}{\emph{Proceedings of the 2016 24th ACM
  SIGSOFT International Symposium on Foundations of Software Engineering
  (FSE)}}. \bibinfo{pages}{193–204}.
\newblock


\bibitem[Donz{\'{e}} et~al\mbox{.}(2014)]%
        {Donze2014}
\bibfield{author}{\bibinfo{person}{Alexandre Donz{\'{e}}},
  \bibinfo{person}{Rafael Valle}, \bibinfo{person}{Ilge Akkaya},
  \bibinfo{person}{Sophie Libkind}, \bibinfo{person}{Sanjit~A. Seshia}, {and}
  \bibinfo{person}{David Wessel}.} \bibinfo{year}{2014}\natexlab{}.
\newblock \showarticletitle{Machine Improvisation with Formal Specifications}.
  In \bibinfo{booktitle}{\emph{Music Technology meets Philosophy - From Digital
  Echos to Virtual Ethos: Joint Proceedings of the 40th International Computer
  Music Conference, {ICMC}}}.
\newblock
\urldef\tempurl%
\url{https://hdl.handle.net/2027/spo.bbp2372.2014.196}
\showURL{%
\tempurl}


\bibitem[Gao et~al\mbox{.}(2019)]%
        {Gao2019}
\bibfield{author}{\bibinfo{person}{Pengfei Gao}, \bibinfo{person}{Jun Zhang},
  \bibinfo{person}{Fu Song}, {and} \bibinfo{person}{Chao Wang}.}
  \bibinfo{year}{2019}\natexlab{}.
\newblock \showarticletitle{Verifying and Quantifying Side-Channel Resistance
  of Masked Software Implementations}.
\newblock \bibinfo{journal}{\emph{ACM Trans. Softw. Eng. Methodol.}}
  \bibinfo{volume}{28}, \bibinfo{number}{3}, Article \bibinfo{articleno}{16}
  (\bibinfo{date}{jul} \bibinfo{year}{2019}), \bibinfo{numpages}{32}~pages.
\newblock
\showISSN{1049-331X}
\urldef\tempurl%
\url{https://doi.org/10.1145/3330392}
\showDOI{\tempurl}


\bibitem[Gore et~al\mbox{.}(1997)]%
        {GoreJKSM97}
\bibfield{author}{\bibinfo{person}{Vivek Gore}, \bibinfo{person}{Mark Jerrum},
  \bibinfo{person}{Sampath Kannan}, \bibinfo{person}{Z. Sweedyk}, {and}
  \bibinfo{person}{Stephen~R. Mahaney}.} \bibinfo{year}{1997}\natexlab{}.
\newblock \showarticletitle{A Quasi-Polynomial-Time Algorithm for Sampling
  Words from a Context-Free Language}.
\newblock \bibinfo{journal}{\emph{Inf. Comput.}} \bibinfo{volume}{134},
  \bibinfo{number}{1} (\bibinfo{year}{1997}), \bibinfo{pages}{59--74}.
\newblock
\urldef\tempurl%
\url{https://doi.org/10.1006/inco.1997.2621}
\showDOI{\tempurl}


\bibitem[Iverson(1962)]%
        {Ive62}
\bibfield{author}{\bibinfo{person}{Kenneth~E Iverson}.}
  \bibinfo{year}{1962}\natexlab{}.
\newblock \showarticletitle{A programming language}. In
  \bibinfo{booktitle}{\emph{Proceedings of the May 1-3, 1962, spring joint
  computer conference}}. \bibinfo{pages}{345--351}.
\newblock


\bibitem[Jerrum et~al\mbox{.}(1986)]%
        {JerrumVV86}
\bibfield{author}{\bibinfo{person}{Mark Jerrum}, \bibinfo{person}{Leslie~G.
  Valiant}, {and} \bibinfo{person}{Vijay~V. Vazirani}.}
  \bibinfo{year}{1986}\natexlab{}.
\newblock \showarticletitle{Random Generation of Combinatorial Structures from
  a Uniform Distribution}.
\newblock \bibinfo{journal}{\emph{Theor. Comput. Sci.}}  \bibinfo{volume}{43}
  (\bibinfo{year}{1986}), \bibinfo{pages}{169--188}.
\newblock
\urldef\tempurl%
\url{https://doi.org/10.1016/0304-3975(86)90174-X}
\showDOI{\tempurl}


\bibitem[Kannan et~al\mbox{.}(1995)]%
        {KannanSM95}
\bibfield{author}{\bibinfo{person}{Sampath Kannan}, \bibinfo{person}{Z.
  Sweedyk}, {and} \bibinfo{person}{Stephen~R. Mahaney}.}
  \bibinfo{year}{1995}\natexlab{}.
\newblock \showarticletitle{Counting and Random Generation of Strings in
  Regular Languages}. In \bibinfo{booktitle}{\emph{Proceedings of the Sixth
  Annual {ACM-SIAM} Symposium on Discrete Algorithms (SODA)}}.
  \bibinfo{pages}{551--557}.
\newblock


\bibitem[Karp and Luby(1985)]%
        {KarpL85}
\bibfield{author}{\bibinfo{person}{Richard~M. Karp} {and}
  \bibinfo{person}{Michael Luby}.} \bibinfo{year}{1985}\natexlab{}.
\newblock \showarticletitle{Monte-Carlo algorithms for the planar multiterminal
  network reliability problem}.
\newblock \bibinfo{journal}{\emph{J. Complex.}} \bibinfo{volume}{1},
  \bibinfo{number}{1} (\bibinfo{year}{1985}), \bibinfo{pages}{45--64}.
\newblock
\urldef\tempurl%
\url{https://doi.org/10.1016/0885-064X(85)90021-4}
\showDOI{\tempurl}


\bibitem[Legay et~al\mbox{.}(2010)]%
        {SMC2010}
\bibfield{author}{\bibinfo{person}{Axel Legay}, \bibinfo{person}{Beno{\^i}t
  Delahaye}, {and} \bibinfo{person}{Saddek Bensalem}.}
  \bibinfo{year}{2010}\natexlab{}.
\newblock \showarticletitle{Statistical Model Checking: An Overview}. In
  \bibinfo{booktitle}{\emph{Runtime Verification}}. \bibinfo{pages}{122--135}.
\newblock


\bibitem[Ozkan et~al\mbox{.}(2019)]%
        {Ozkan2019}
\bibfield{author}{\bibinfo{person}{Burcu~Kulahcioglu Ozkan},
  \bibinfo{person}{Rupak Majumdar}, {and} \bibinfo{person}{Simin Oraee}.}
  \bibinfo{year}{2019}\natexlab{}.
\newblock \showarticletitle{Trace Aware Random Testing for Distributed
  Systems}.
\newblock \bibinfo{journal}{\emph{Proc. ACM Program. Lang.}}
  \bibinfo{volume}{3}, \bibinfo{number}{OOPSLA}, Article
  \bibinfo{articleno}{180} (\bibinfo{date}{oct} \bibinfo{year}{2019}),
  \bibinfo{numpages}{29}~pages.
\newblock
\urldef\tempurl%
\url{https://doi.org/10.1145/3360606}
\showDOI{\tempurl}


\bibitem[Saha et~al\mbox{.}(2023)]%
        {Saha2023}
\bibfield{author}{\bibinfo{person}{Seemanta Saha}, \bibinfo{person}{Surendra
  Ghentiyala}, \bibinfo{person}{Shihua Lu}, \bibinfo{person}{Lucas Bang}, {and}
  \bibinfo{person}{Tevfik Bultan}.} \bibinfo{year}{2023}\natexlab{}.
\newblock \showarticletitle{Obtaining Information Leakage Bounds via
  Approximate Model Counting}.
\newblock \bibinfo{journal}{\emph{Proc. ACM Program. Lang.}}
  \bibinfo{volume}{7}, \bibinfo{number}{PLDI}, Article \bibinfo{articleno}{167}
  (\bibinfo{date}{jun} \bibinfo{year}{2023}).
\newblock
\urldef\tempurl%
\url{https://doi.org/10.1145/3591281}
\showDOI{\tempurl}


\bibitem[Sutton et~al\mbox{.}(2007)]%
        {sutton2007fuzzing}
\bibfield{author}{\bibinfo{person}{Michael Sutton}, \bibinfo{person}{Adam
  Greene}, {and} \bibinfo{person}{Pedram Amini}.}
  \bibinfo{year}{2007}\natexlab{}.
\newblock \bibinfo{booktitle}{\emph{Fuzzing: Brute Force Vulnerability
  Discovery}}.
\newblock \bibinfo{publisher}{Addison-Wesley Professional}.
\newblock
\showISBNx{0321446119}


\bibitem[van Bremen and Meel(2023)]%
        {vanBremen2023}
\bibfield{author}{\bibinfo{person}{Timothy van Bremen} {and}
  \bibinfo{person}{Kuldeep~S. Meel}.} \bibinfo{year}{2023}\natexlab{}.
\newblock \showarticletitle{Probabilistic Query Evaluation: The Combined FPRAS
  Landscape}. In \bibinfo{booktitle}{\emph{PODS}}. \bibinfo{pages}{339–347}.
\newblock


\bibitem[Álvarez and Jenner(1993)]%
        {ALVAREZ19933}
\bibfield{author}{\bibinfo{person}{Carme Álvarez} {and}
  \bibinfo{person}{Birgit Jenner}.} \bibinfo{year}{1993}\natexlab{}.
\newblock \showarticletitle{A very hard log-space counting class}.
\newblock \bibinfo{journal}{\emph{Theoretical Computer Science}}
  \bibinfo{volume}{107}, \bibinfo{number}{1} (\bibinfo{year}{1993}),
  \bibinfo{pages}{3--30}.
\newblock
\showISSN{0304-3975}
\urldef\tempurl%
\url{https://doi.org/10.1016/0304-3975(93)90252-O}
\showDOI{\tempurl}


\end{thebibliography}

\clearpage
\appendix



\section{Details from \secref{mainalgo}}

\subsection{Joint Distribution on $X$ and $Y$}
\applabel{joint-distribution}

Here, we will outline a concrete joint distribution that satisfies the
properties we describe in \secref{mainalgo}.
We will present the definition inductively, 
first defining the slice $\set{X^{\stateq}_1, Y^{\stateq}_1}_{\stateq \in \states}$, 
and then use the slices $\set{X^{\stateq}_j, Y^{\stateq}_j}_{\stateq \in \states, 0 \leq j \leq i}$
to define the slice $\set{X^{\stateq}_{i+1}, Y^{\stateq}_{i+1}}_{\stateq \in \states}$.
Our inductive definition will also allow us to establish the properties we outlined.
In the following, we use the notation $\Omega_{X, i, \stateq}$ (resp. $\Omega_{Y, i, \stateq}$)
to denote the domain of $X^{\stateq}_i$ (resp. $Y^{\stateq}_i$).
Analogously, we will use the notations $\Omega_{X, i}$, $\Omega_{Y, i}$,
$\Omega_{X, \leq i}$ and $\Omega_{Y, \leq i}$
to denote respectively
$\Pi_{\stateq \in \states} \Omega_{X, i, \stateq}$,
$\Pi_{\stateq \in \states} \Omega_{Y, i, \stateq}$,
$\Pi_{j=1}^i \Omega_{X, j}$, and
$\Pi_{j=1}^i \Omega_{Y, j}$. \\

\begin{description}

\item[Base Case(i=0).] Follows from the definition.
\item[Inductive Case.]
We will define the distribution inductively
by defining the conditional probabilities
$\prob{X^{\stateq}_{i+1} = \omega_{X, i+1, \stateq} \cap  Y^{\stateq}_{i+1} = \omega_{Y, i+1, \stateq} | X_{\leq i} = \omega_{X, \leq i} \cap  Y_{\leq i} = \omega_{Y, \leq i} \cap  \accurateEstevent_{\leq i}}$
and 

$\prob{X^{\stateq}_{i+1} = \omega_{X, i+1, \stateq} \cap  Y^{\stateq}_{i+1} = \omega_{Y, i+1, \stateq} | X_{\leq i} = \omega_{X, \leq i} \cap  Y_{\leq i} = \omega_{Y, \leq i} \cap  \overline{\accurateEstevent_{\leq i}}}$.

For each $A \in \set{\accurateEstevent_{\leq i}, \overline{\accurateEstevent_{\leq i}}}$, we have:
\begin{align*}
\begin{array}{l}
\prob{X^{\stateq}_{i+1} = \omega_{X, i+1, \stateq} \cap  Y^{\stateq}_{i+1} = \omega_{Y, i+1, \stateq} | X_{\leq i} = \omega_{X, \leq i} \cap  Y_{\leq i} = \omega_{Y, \leq i} \cap  A} \\
\quad\quad =
\begin{cases}
	\min\left(
	\begin{array}{l}
		\prob{\usample{\stlvl{\stateq}{i+1}} = \omega_{X, i+1, \stateq}}, \\
		\prob{\sample{\stlvl{\stateq}{i+1}} = \omega_{X, i+1, \stateq} | {\sf S}_{\leq i} = \omega_{X, \leq i} \cap  A}
	\end{array}
	\right)
	& 
	\text{if } \omega_{X, i+1, \stateq} = \omega_{Y, i+1, \stateq} \\\\
	\begin{array}{c}
		\left[
		\begin{array}{l}
			\prob{\sample{\stlvl{\stateq}{i+1}} = \omega_{X, i+1, \stateq} | {\sf S}_{\leq i} = \omega_{X, \leq i} \cap  A} \\
			- \\
			\min\left(
				\begin{array}{l}
					\prob{\usample{\stlvl{\stateq}{i+1}} = \omega_{X, i+1, \stateq}}, \\
					\prob{\sample{\stlvl{\stateq}{i+1}} = \omega_{X, i+1, \stateq} | {\sf S}_{\leq i} = \omega_{X, \leq i} \cap  A}
				\end{array}
				\right)
		\end{array}
		\right] \\
		\hline
		|\Omega_{Y, i+1, \stateq}|-1
	\end{array}
	&
	\text{otherwise}
\end{cases}
\end{array}
\end{align*}

\noindent
	Let us first observe that for each $A\in \set{\accurateEstevent_{\leq i}, \overline{\accurateEstevent_{\leq i}}}$ and $\omega_{X, \leq i}$, we have
	\begin{align*}
	\prob{X_{i+1}^{\stateq} = \omega_{X, i+1, \stateq} | {X}_{\leq i} = \omega_{X, \leq i} \cap  A} \\
	= 
	&
	\prob{\sample{\stlvl{\stateq}{i+1}} = \omega_{X, i+1, \stateq} | {\sf S}_{\leq i} = \omega_{X, \leq i} \cap  A}
	\end{align*}

Also, it is clear that $\sum_{\omega_{X, \leq i+1}, \omega_{Y, \leq i+1}}\prob{X_{\leq i+1} = \omega_{X, \leq i+1} \cap  Y_{\leq i+1} = \omega_{Y, \leq i+1}} = 1$.
\end{description}

It is easy to see that the above joint distribution satisfies both the desired properties.

\subsection{Proof of Claim~\ref{claim:xbound}}
\xboundclaim*
\begin{proof}
	\begin{align*}
		& \prob{(X_{\leq \ell-1} =Y_{\leq \ell-1} = \omega) \cap \accurateEstevent_{\leq \lvl-1}} \\
		= & \prob{(X_{1} =Y_{1} = \omega_1) \cap \accurateEstevent_{1}}\times \\
		& \quad  \prob{(X_{2} =Y_{2} = \omega_2) \cap \accurateEstevent_{2} \mid (X_{1} =Y_{1} = \omega_1) \cap \accurateEstevent_{1}}\times \\
		& \quad \quad \prob{(X_{3} =Y_{3} = \omega_3) \cap \accurateEstevent_{3} \mid (X_{\leq 2} =Y_{\leq 2} = \omega_1\omega_2) \cap \accurateEstevent_{\leq 2}}\times \dots \\
		= & \prod_{k=1}^{\ell} \prob{(X_{k} =Y_{k} = \omega_k) \cap \accurateEstevent_{k} \mid (X_{\leq (k-1)} =Y_{\leq (k-1)} = \omega') \cap \accurateEstevent_{\leq (k-1)}}
	\end{align*}
	Now, since $\Pr[A\cap B \mid C] = \frac{\Pr[A\cap B \cap C]}{\Pr[C]} \leq  \frac{\Pr[A\cap B \cap C]}{\Pr[B\cap C]} = \Pr[A\mid B\cap C]$ so,
	\begin{align*} 
		&  \prob{(X_{k} =Y_{k} = \omega_k) \cap \accurateEstevent_{k} \mid (X_{\leq (k-1)} =Y_{\leq (k-1)} = \omega') \cap \accurateEstevent_{\leq (k-1)}}\\
		\leq &  \prob{(X_{k} =Y_{k} = \omega_k) \mid (X_{\leq (k-1)} =Y_{\leq (k-1)} = \omega') \cap \accurateEstevent_{\leq (k)}}
	\end{align*}
	
	So, 
	\begin{align*} 
		&  \prob{(X_{k} =Y_{k} = \omega_k) \cap \accurateEstevent_{k} \mid (X_{\leq (k-1)} =Y_{\leq (k-1)} = \omega') \cap \accurateEstevent_{\leq (k-1)}}\\
		\leq & \prob{(X_{k} =Y_{k} = \omega_k) \mid (X_{\leq (k-1)} =Y_{\leq (k-1)} = \omega') \cap \accurateEstevent_{\leq (k)}}
	\end{align*}
	
	From the second condition of the joint random variables $X$ and $Y$ we have 
	\begin{align*}
		\prob{(X_{k} =Y_{k} = \omega_k) \mid (X_{\leq (k-1)} =Y_{\leq (k-1)} = \omega') \cap \accurateEstevent_{\leq (k)}} \leq & \prob{U_{\leq k} = \omega}
	\end{align*}
	So we have, 
	\begin{align*}
		& \prob{(X_{\leq \ell-1} =Y_{\leq \ell-1} = \omega) \cap \accurateEstevent_{\leq \lvl-1}} \\
		= & \prod_{k=1}^{\ell} \prob{(X_{k} =Y_{k} = \omega_k) \cap \accurateEstevent_{k} \mid (X_{\leq (k-1)} =Y_{\leq (k-1)} = \omega') \cap \accurateEstevent_{\leq (k-1)}}\\
		\leq & \prod_{k=1}^{\ell} \prob{(X_{k} =Y_{k} = \omega_k) \mid (X_{\leq (k-1)} =Y_{\leq (k-1)} = \omega') \cap \accurateEstevent_{\leq (k)}}\\
		\leq &  \prod_{k=1}^{\ell} \prob{U_{\leq k} = \omega}  =  \prob{U_{\leq \ell-1} = \omega},
	\end{align*}
	the last equality since the random variables $U_k$ are are uniformly generated and hence independent. 
\end{proof}

\subsection{Proof of Claim~\ref{claim:fermat}}

\accevent*
\begin{proof}
	\begin{align*}
		\overline{\accurateEstevent_{n}} \subseteq 
		&(\overline{\accurateEstevent_{n}} \cap (X_{\leq n-1} = Y_{\leq n-1}) ) 
		\cup (X_{\leq n-1} \neq Y_{\leq n-1}) 
	\end{align*}
	Expanding further, we have 
	\begin{align*}
		\overline{\accurateEstevent_{n}}\subseteq 
		&(\overline{\accurateEstevent_{n}} \cap 
		{\accurateEstevent_{\leq n-1}} \cap
		(X_{\leq n-1} = Y_{\leq n-1}) ) \\ 
		&\cup (\overline{\accurateEstevent_{n}} \cap 
		\overline{\accurateEstevent_{\leq n-1}} \cap
		(X_{\leq n-1} = Y_{\leq n-1}) )  \\
		& \cup ((X_{\leq n-1} \neq Y_{\leq n-1}) \cap \overline{\accurateEstevent_{\leq n-1}} )  \cup ((X_{\leq n-1} \neq Y_{\leq n-1}) \cap {\accurateEstevent_{\leq n-1}} ) 
	\end{align*}

	Observing 
	$ (\overline{\accurateEstevent_{n}} \cap 
	\overline{\accurateEstevent_{\leq n-1}} )  \subseteq \overline{\accurateEstevent_{\leq n-1}}$ and then combining second and third terms, we have
	\begin{align*}
		\overline{\accurateEstevent_{n}}
		\subseteq 
		&(\overline{\accurateEstevent_{n}} \cap 
		{\accurateEstevent_{\leq n-1}} \cap
		(X_{\leq n-1} = Y_{\leq n-1}) ) \cup \overline{\accurateEstevent_{\leq n-1}}  \\
		& \cup ((X_{\leq n-1} \neq Y_{\leq n-1}) \cap {\accurateEstevent_{\leq n-1}} ) 
	\end{align*}
	
	Observe that $\overline{\accurateEstevent_{\leq n-1}} = \bigcup_{\ell=1}^{n-1}  \overline{\accurateEstevent_{\ell}}$. Therefore, we can further applying the above recurrence, in particular, for every $\ell$, we will use the following upper bound. 
	\begin{align*}
		\overline{\accurateEstevent_{\ell}} \subseteq 
		&(\overline{\accurateEstevent_{\ell}} \cap (X_{\leq \ell-1} = Y_{\leq \ell-1}) )  \cup (X_{\leq \ell-1} \neq Y_{\leq \ell-1}) 
	\end{align*}
	\begin{align*}
		\mbox{Therefore, } \quad \overline{\accurateEstevent_{n}} \subseteq 
		&\bigcup_{\ell=1}^{n}(\overline{\accurateEstevent_{\ell}} \cap 
		{\accurateEstevent_{\leq \ell-1}} \cap
		(X_{\leq \ell-1} = Y_{\leq \ell-1}) ) \\ 
		&\qquad\qquad \cup \bigcup_{\ell=1}^{n-1} ((X_{\leq \ell-1} \neq Y_{\leq \ell-1}) \cap {\accurateEstevent_{\leq \ell-1}} ) 
	\end{align*}

	Let us focus on two successive terms of the second expression 
	and using the equation $A = A \cup (\overline{A} \cap B)$
	\begin{align*}
		&
		\underbrace{\left((X_{\leq \ell-1} \neq Y_{\leq \ell-1}) \cap {\accurateEstevent_{\leq \ell-1}} \right)}_{A}
		\cup 	
		\underbrace{\left((X_{\leq \ell} \neq Y_{\leq \ell}) \cap {\accurateEstevent_{\leq \ell}}  \right)}_{B} \\
		= & ((X_{\leq \ell-1} \neq Y_{\leq \ell-1}) \cap {\accurateEstevent_{\leq \ell-1}} )  \cup 
		((X_{\leq \ell-1} = Y_{\leq \ell-1}) \cap (X_{\leq \ell} \neq Y_{\leq \ell}) \cap {\accurateEstevent_{\leq \ell}} ) \\ 
		&\quad\quad \cup 
		(\overline{\accurateEstevent_{\leq \ell-1}} \cap (X_{\leq \ell} \neq Y_{\leq \ell}) \cap {\accurateEstevent_{\leq \ell}} ) \\
	\end{align*}
	
	Noting that $((X_{\leq \ell-1} = Y_{\leq \ell-1}) \cap (X_{\leq \ell} \neq Y_{\leq \ell}) = ((X_{\leq \ell-1} = Y_{\leq \ell-1}) \cap (X_{\ell} \neq Y_{\ell})$ and the fact that
	$(\overline{\accurateEstevent_{\leq \ell-1}}  \cap {\accurateEstevent_{\leq \ell}} ) = \emptyset $ we have,
	\begin{align*}
		&((X_{\leq \ell-1} \neq Y_{\leq \ell-1}) \cap {\accurateEstevent_{\leq \ell-1}} \cup 	((X_{\leq \ell} \neq Y_{\leq \ell}) \cap {\accurateEstevent_{\leq \ell}} ) \\= 
		& ((X_{\leq \ell-1} \neq Y_{\leq \ell-1}) \cap {\accurateEstevent_{\leq \ell-1}} )  \cup 
		\left((X_{ \ell} \neq Y_{ \ell}) \cap (X_{\leq \ell-1} = Y_{\leq \ell-1}) \cap  {\accurateEstevent_{\leq \ell}}\right) 
	\end{align*}
	
	Applying the above simplification for all terms, 
	
	\begin{align*}
		& \bigcup_{\ell=1}^{n-1} ((X_{\leq \ell-1} \neq Y_{\leq \ell-1}) \cap {\accurateEstevent_{\leq \ell-1}} ) =  \bigcup_{\ell=1}^{n-1} \left(
		(X_{ \ell} \neq Y_{ \ell}) \cap
		(X_{\leq \ell-1} = Y_{\leq \ell-1}) \cap
		{\accurateEstevent_{\leq \ell}} 
		\right) \\
		&\text{Therefore, }	\overline{\accurateEstevent_{n}}  \subseteq  
		\bigcup_{\ell=1}^{n-1}(\overline{\accurateEstevent_{\ell}} \cap 
		{\accurateEstevent_{\leq \ell-1}} \cap
		(X_{\leq \ell-1} = Y_{\leq \ell-1}) ) \\ 
		& \qquad\qquad\qquad\qquad\qquad\qquad\qquad\quad \cup \bigcup_{\ell=1}^{n-1} \left(
		(X_{ \ell} \neq Y_{ \ell}) \cap
		(X_{\leq \ell-1} = Y_{\leq \ell-1}) \cap
		{\accurateEstevent_{\leq \ell}} 
		\right) 
	\end{align*}
	
\end{proof}

\end{document}